\documentclass[journal,10pt]{IEEEtran}

\makeatletter
\def\ps@headings{%
\def\@oddhead{\mbox{}\scriptsize\rightmark \hfil \thepage}%
\def\@evenhead{\scriptsize\thepage \hfil \leftmark\mbox{}}%
\def\@oddfoot{}%
\def\@evenfoot{}}
\makeatother
\usepackage{amsfonts}
\usepackage{mathrsfs}
\usepackage{amsfonts}
\usepackage{amssymb}
\usepackage{graphicx}
\usepackage{epsfig}
\usepackage{psfrag}
\usepackage{amsmath}
\usepackage{array}
\usepackage{cases}
\usepackage{subfigure}
\usepackage{cite,graphicx,amsmath,amssymb,color}
\usepackage{anysize}
\usepackage{algorithmic}
\usepackage{algorithm}

\usepackage{algorithmic}
\usepackage{stmaryrd}
\usepackage{multirow}
\usepackage{subfig}
\usepackage{graphicx,times,amsmath} 

\usepackage{url}

\marginsize{13mm}{13mm}{15mm}{20mm}

\begin{document}

\allowdisplaybreaks

\title{A  Linear Network Code Construction for \\ General Integer Connections Based on \\ the Constraint Satisfaction Problem}

\author{Ying Cui,\thanks{Y. Cui and F. Lai are  with Shanghai Jiao Tong University. M. M\'{e}dard is  with MIT. E. Yeh is with Northeastern University. D. Leith is with Trinity College Dublin. K. Duffy is with Maynooth University.  D. Pandya is with  Harvard University.  The paper was presented in part in GLOBECOM 2015.} \ Muriel M\'{e}dard, \ Fan Lai, \  Edmund Yeh, \ Douglas Leith, \ Ken Duffy,\ Dhaivat Pandya}



\maketitle

\newtheorem{Prop}{Proposition}
\newtheorem{Thm}{Theorem}
\newtheorem{Lem}{Lemma}
\newtheorem{Cor}{Corollary}
\newtheorem{Def}{Definition}
\newtheorem{Exam}{Example}
\newtheorem{Alg}{Algorithm}
\newtheorem{Sch}{Scheme}
\newtheorem{Prob}{Problem}
\newtheorem{Rem}{Remark}
\newtheorem{Proof}{Proof}
\newtheorem{Asump}{Assumption}
\newtheorem{Subp}{Subproblem}

\begin{abstract}  The problem of finding network codes for general connections is  inherently difficult  in capacity constrained networks. 
Resource minimization for general connections with network coding  is further complicated.
Existing methods for identifying solutions mainly rely on  highly restricted classes of network codes, and are almost all centralized. In this paper,  we introduce linear network mixing coefficients for code constructions of general connections that generalize random linear network coding (RLNC) for multicast connections. For such code constructions, we pose the problem of cost minimization for the subgraph involved in the coding solution and relate this minimization to a path-based  Constraint Satisfaction Problem (CSP) and an edge-based CSP. While CSPs are NP-complete in general, we
present a path-based probabilistic distributed algorithm and an edge-based probabilistic distributed algorithm with almost sure convergence in finite time  by  applying Communication Free Learning (CFL).
Our approach allows fairly general coding across flows, guarantees no greater cost than routing,    and  shows a possible distributed implementation.
Numerical results illustrate the performance improvement of our approach over  existing methods.
\end{abstract}



\section{Introduction}\label{sec:intro}

The problem of finding network codes in the case of general connections, where each destination can request information from any subset of sources, is intrinsically difficult and little is known about its complexity.
In certain special cases, such as multicast connections (where destinations share all of their demands), it suffices to satisfy a Ford-Fulkerson type of min-cut max-flow constraint between all sources to every destination individually.
For multicast connections, linear codes suffice \cite{KM03, LYC03}, and lend themselves to a distributed random construction \cite{Hoetal06}.  While linear codes have been the most widely considered in the literature, linear codes over finite fields may in general not be sufficient for general connections, as shown by \cite{DFZ05} using an example from matroid theory.


A matroidal structure for the network coding problem with general connections was conjectured by the late Ralf K\"otter (private communication) but, while different aspects of this connection have been  investigated in the literature \cite{ESG09, SHL08, DFZ07, KM10, YYZ, CGP11, SMISIT14},
a proof  remains elusive, except in special cases.
 Recently, the problem of scalar-linear coding has been shown to have a matroidal structure \cite{DFZ07, DFZ07IT, KM10}.
There exists a correspondence between scalar-linearly solvable networks and representable matroids over finite fields, which can be used to obtain some bounds on scalar linear network capacity \cite{KM11} or the capacity regions of certain classes of networks \cite{DFZ12}.
More generally, the problem of finding the linear network coding capacity
region is equivalent to the characterization of all linear polymatroids \cite{YYZ}, whose structure  was investigated in \cite{CGP11}.
Reference \cite{SMISIT14}  generalized the results of \cite{ESG10}, which investigated the connection among index coding, network coding and matroid theory.   In \cite{SMCAllerton2015},  polymatroids  were used to produce linear code constructions.

Progress in understanding the matroidal structure of the general connection problem has, however, not yet provided simple and useful approaches to generating explicit linear codes.
There has been considerable investigation of restricted cases, such as
a network with only  two sources and two destinations, generally referred to as
 the two-unicast network \cite{ZCM12-1, WS10, KTA11,  KamathATW15, ZCM14}, but thus far such investigation has yielded only bounds or linear solutions for restricted cases of the two-unicast network.  It has been  shown in \cite{KamathATW15} that the two-unicast problem is as hard as the most general network coding problem. 
 Since the difficulty of coding in the case of general connections is in effect an interference cancellation one, approaches relying on interference alignment have naturally been explored \cite{MDRJMV13, ZCM13, MCJ12}.  Reference \cite{LiWW15} investigated the enumeration, rate region computation and hierarchy of general multi-source multi-sink
hyper-edge networks under network coding.

Even when we consider simple scalar network codes, which have scalar coding coefficients, the problem of code construction  for general connections  remains vexing. 
The main difficulty lies in cancelling the effect of flows that are coded together even though they are not destined for a common destination.
The problem of code construction is further complicated when we seek, for common reasons of network resource management, to limit fully or partially  the use of links in the network.
For convex cost functions of  flows over edges in the graph corresponding to the network, finding a minimum-cost  solution is known to be a convex optimization problem in the case of multicast connections (for continuous flows)\cite{Lunetal06}.
However, in the case of general connections, network resource minimization, even when allowing only restricted code constructions, appears difficult.

Among coding approaches for optimizing network use for general connections, we distinguish two types. The first,  which we adopt in this paper, is that of mixing,  by which we mean
coding together flows  using  random linear network coding (RLNC)~\cite{Hoetal06}, originally proposed for multicast connections.  The principle  is to code together flows as though they were part of a common multicast connection.
In this case, no explicit coding coefficients are provided, and decidability is ensured with high probability by  RLNC.
For example,  the mixing approaches in \cite{Lun04networkcoding} and \cite{WuISIT2006} are both based on mixing variables, each corresponding to a set of flows that can be mixed over an edge.
Specifically, in \cite{Lun04networkcoding}, a two-step mixing approach is proposed for network resource minimization of general connections, where flow partition (mixing) and flow rate optimization are considered separately.  This separation imposes stronger restrictions on the mixing  design in the first step and leads to a limitation on the feasibility region. Reference  \cite{WuISIT2006} studies the  feasibility of  more general mixing designs based on  mixing variables of size $\mathcal O(2^P)$, where $P$ is the number of flows. Reference \cite{WuISIT2006} does not, however, provide an approach for obtaining a specific mixing design.
The second type of coding approach is an explicit linear code construction, by which we mean providing specific linear coefficients over some finite field, to be applied to  coding flows at different nodes. Often these constructions are simplified by restricting them to be pairwise. 
For example, in \cite{TRKLM06} and \cite{ShroffBButterflyrate},   simple codes over pairs of flows   are proposed for network resource minimization of  general connections.

 Some explicit linear network code construction approaches \cite{TRKLM06}, \cite{ShroffBButterflyrate} are distributed, but they allow only pairwise coding.
 The algorithms of \cite{KOMT09} using evolutionary techniques, which are also explicit code constructions, are partially distributed, since
the chromosomes can be decomposed into their local contributions, but  require  
information  to be fed back from the receivers to all the nodes in the network.   In addition, the convergence results for  evolutionary techniques are generally scant   and do not yield prescriptive constructions. 
While RLNC for multicast connections is a distributed algorithm, most of the mixing approaches \cite{Lun04networkcoding}, \cite{WuISIT2006} based on it have remained centralized.
In \cite{CUI2015ICC}, we propose new methods for constructing linear network codes for general connections of continuous flows based on mixing to minimize the total network cost. Flow splitting and coding over time are required to achieve the desired performance. The focus in \cite{CUI2015ICC} is to apply continuous optimization techniques to obtain continuous flow rates. In \cite{CUI2015GLOBECOM}, we consider linear network code construction for general connections of integer flows based on mixing, and propose an edge-based probabilistic distributed algorithm  to minimize the total network cost. This paper extends the results in \cite{CUI2015GLOBECOM}.  

Our contribution  in this paper is to present new methods for constructing linear network codes  in a distributed manner for general connections of integer flows based on mixing.

$\bullet$  We introduce linear network mixing coefficients.  The number of mixing coefficients grows   polynomially with the number of flows. We formally establish the relationship between linear network coding and mixing.

$\bullet$ We formulate   the minimization of the cost of the subgraph involved in the code construction for general  connections of integer flows in terms of  the mixing coefficients.

$\bullet$ We relate our problem to a  path-based  Constraint Satisfaction Problem  (CSP) and an edge-based CSP. While CSPs are NP-complete in general, we
present a path-based probabilistic distributed algorithm and an edge-based probabilistic distributed algorithm with almost sure convergence in finite time by  applying  Communication Free Learning (CFL), a recent probabilistic distributed solution for CSPs\cite{cfl}. The path-based distributed algorithm requires more local information than the edge-based distributed algorithm, but converges faster.



$\bullet$ We show that our approach  guarantees no greater cost than routing  or the  simplified mixing design in \cite{Lun04networkcoding}. Numerical  results also illustrate the performance improvement of our approach over   existing methods.

While our approach, like all other general connection code constructions, is generally suboptimal, it allows more flows to be mixed than is possible with pairwise mixing \cite{TRKLM06}, \cite{ShroffBButterflyrate} and with the separate mixing design in \cite{Lun04networkcoding}.   Moreover,   in contrast to  \cite{Lun04networkcoding,WuISIT2006,CUI2015ICC}, our approach does not require  non-scalar  coding over time.

\section{Problem Setup and Definitions}


\subsection{Network Model}

We consider a directed acyclic network with general connections.\footnote{The network model   we considered in this paper is similar to that in \cite{CUI2015ICC} for continuous flows, but here we consider integer flows and edge capacities, and  do not allow flow splitting and coding over time.} Let $\mathcal G=(\mathcal V, \mathcal E)$ denote the  directed acyclic  graph, where $\cal V$ denotes the set of $V=|\cal V|$ nodes and $\cal E$ denotes the set of $E=|\cal E|$ edges.  To simplify notation, we assume there is only one edge from  node $i\in\mathcal V$ to node $j\in \mathcal V$, denoted as edge $(i,j)\in\mathcal E$.\footnote{Multiple edges from node $i$ to node $j$ can be modeled by introducing multiple extra nodes,  one on each edge, to transform a multigraph intro a graph.}
For each node $i\in \cal V$, define the set of incoming neighbors to be $\mathcal I_i=\{j: (j,i)\in \mathcal E \}$ and the set of outgoing neighbors to be $\mathcal O_i=\{j:(i,j)\in \mathcal E\}$.  Let $I_i=|\mathcal I_i|$ and $O_i=|\mathcal O_i|$ denote the in-degree and out-degree of node $i\in \mathcal V$, respectively. Assume $I_i \leq D$ and $O_i\leq D$ for all $i \in \mathcal V$, where $D$ is a constant.

Consider a  finite field $\mathcal F$ with size $F=|\mathcal F|$. 
Let $\mathcal P=\{1,\cdots, P\}$ denote the set of $P=|\mathcal P|$ flows  of symbols in finite field $\mathcal F$  to be carried by the network.  For each flow $p\in \mathcal P$, let $s_p\in \mathcal V$ be its source.  We consider integer flows. To simplify notation,   we assume unit source rate (i.e., one finite field symbol per second).\footnote{A source with  a positive integer source rate greater than one can be  modeled by multiple sources, each with unit source rate.}
Let $\mathcal S=\{s_1,\cdots, s_P\}$ denote the set of $P=|\mathcal S|$ sources.
We assume different flows do not share a common source node and no source node has any incoming edges.
Let $\mathcal T=\{t_1,\cdots, t_T\}$ denote the set of $T=|\mathcal T|$ terminals. Each terminal $t\in \mathcal T$ demands a subset of $P_t=|\mathcal P_t|$ flows $\mathcal P_t\subseteq \mathcal P$. Assume $\cup_{t\in\mathcal T}\mathcal P_t=\mathcal P$.  Let $\boldsymbol{\mathcal P}\triangleq (\mathcal P_t)_{t\in\mathcal T}$ denote the demands of all the terminals.   We assume no terminal has any outgoing edges.


 As we consider integer flows, we assume unit edge capacity (i.e., one finite field symbol per second).\footnote{An edge with a positive integer edge capacity greater than one can be equivalently converted to multiple edges, each with unit edge capacity.}
 Let $z_{ij}\in \{0,1\}$ denote   whether edge $(i,j)\in \mathcal E$ is  in the subgraph involved in the code construction  in a sense we shall make precise later.\footnote{There is either no flow or a unit rate of (coded) flow through each edge. Under the unit source rate and edge capacity assumptions, we shall see that there is one global coding (mixing) vector for each edge.} We assume a cost is incurred on an edge when information is transmitted through the edge  and let $U_{ij}(z_{ij})$ denote the cost function for edge $(i,j)$.  We assume $U_{ij}(z_{ij})$ is  non-decreasing  in $z_{ij}$.
 We are interested in the problem of finding  linear network coding designs and  minimizing the network cost $\sum_{(i,j)\in \mathcal E} U_{ij}(z_{ij})$ for general connections under those designs.


\subsection{Scalar Time-Invariant Linear Network Coding}\label{subsec:coding}

In linear network coding,   a linear combination  over $\mathcal F$ of the symbols in $\{\sigma_{ki}\in \mathcal F:k\in \mathcal I_i\}$ from the incoming edges $\{(k,i):k\in \mathcal I_i\}$ can be transmitted through the shared edge $(i,j)\in \mathcal E$.  The coefficients used to form this linear combination are referred to as local coding coefficients. Specifically,  let $\alpha_{kij}\in \mathcal F$ denote the local coding coefficient corresponding to edge $(k,i)\in \mathcal E$ and edge $(i,j)\in \mathcal E$.  Denote $\boldsymbol \alpha \triangleq (\alpha_{kij})_{(k,i),(i,j)\in\mathcal E}$. Then, for linear network coding,  using local coding coefficients, the symbol through edge $(i,j)\in \mathcal E$ can be expressed as
\begin{align}
\sigma_{ij}=\sum_{k\in \mathcal I_i}\alpha_{kij}\sigma_{ki}, \quad (i,j)\in \mathcal E, \ i\not\in \mathcal S.\label{eqn:local-coding-coeff}
\end{align}
This is illustrated in Fig.~\ref{Fig:local-global}.  

Starting from the sources, we transmit source symbols $\{\sigma_p\in \mathcal F:p\in \mathcal P\}$, and then, at intermediate nodes, we perform only linear  operations over $\mathcal F$ on the symbols from incoming edges. Thus, the symbol of each edge can be expressed as a linear combination  over $\mathcal F$ of the source symbols $\{\sigma_p\in \mathcal F:p\in \mathcal P\}$. Let $c_{ij,p}\in \mathcal F$ denote the coefficient of flow $p\in \mathcal P$ in the linear combination for edge $(i,j)\in \mathcal E$.
This is referred to as the global coding coefficient of  flow $p\in \mathcal P$ and edge $(i,j)\in \mathcal E$. Let $\mathbf c_{ij}\triangleq (c_{ij,1},\cdots, c_{ij,p},\cdots,c_{ij,P})\in \mathcal F^P$ denote $P$ coefficients corresponding to this linear combination for edge $(i,j)\in \mathcal E$.  This is referred to as the global coding vector of edge $(i,j)\in \mathcal E$. Here, $\mathcal F^P$ represents the set of global coding vectors, the  cardinality  of which is $F^P$.
Then, using global coding vectors, the symbol through edge $(i,j)\in \mathcal E$ can also be  expressed as
\begin{align}
\sigma_{ij}
=\sum_{p\in \mathcal P}c_{ij,p}\sigma_p, \quad (i,j)\in \mathcal E, \  i\not\in \mathcal S.\label{eqn:global-coding-coeff}
\end{align}
This is illustrated in Fig.~\ref{Fig:local-global}.  

\begin{figure}[h]
\begin{center}
\includegraphics[height=2.5cm]{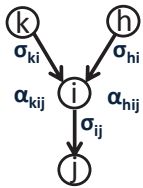}
\caption{\small{Illustration of local and global coding coefficients. $\mathcal P=\{1,2\}$. Then, we have $\sigma_{ki}=c_{ki,1}\sigma_1+c_{ki,2}\sigma_2$, $\sigma_{hi}=c_{hi,1}\sigma_1+c_{hi,2}\sigma_2$, $\sigma_{ij}=\alpha_{kij}\sigma_{ki}+\alpha_{hij}\sigma_{hi}=c_{ij,1}\sigma_1+c_{ij,2}\sigma_2$, $c_{ij,1}=\alpha_{kij}c_{ki,1}+\alpha_{hij}c_{hi,1}$ and $c_{ij,2}=\alpha_{kij}c_{ki,2}+\alpha_{hij}c_{hi,2}$.}}\label{Fig:local-global}
\end{center}
\end{figure}

In this paper, we consider scalar  time-invariant linear network coding.
In other words, $\alpha_{kij}\in \mathcal F$ and $c_{ij,p}\in \mathcal F$ are both scalars, and  do not change over time.
Let $\mathbf e_p$ denote the vector with the $p$-th element being 1 and all the other elements being 0.  For decodability to hold at all the terminals, the global coding vectors at all edges must satisfy the following  feasibility  condition for scalar linear network coding.

\begin{Def} [Feasibility of Scalar Linear Network Coding] For a network $\mathcal G=(\mathcal V, \mathcal E)$ and a set of flows $\mathcal P$ with sources $\mathcal S$  and  terminals $\mathcal T$, a linear network code  $\mathbf c\triangleq (\mathbf c_{ij})_{(i,j)\in\mathcal E}$ 
 is called feasible if the following three conditions are satisfied: 1) $\mathbf c_{s_pj}=\mathbf e_p$ for source edge $(s_p,j)\in \mathcal E$, where  $s_p\in \mathcal S$ and $ p\in \mathcal P$;  2) $\mathbf c_{ij}=\sum_{k\in \mathcal I_i}\alpha_{kij}\mathbf c_{ki}$ for edge  $(i,j) \in \mathcal E$ not outgoing from a source, where $i\not\in \mathcal S$ and $\alpha_{kij}\in \mathcal F$;  3)   $\mathbf e_p\in \text{span}\{\mathbf c_{it}: i\in \mathcal I_t\}$, where $ p\in \mathcal P_t$ and $t\in \mathcal T$.
\label{Def:feasibility}
\end{Def}

 Note that when using scalar linear network coding, for each terminal,  extraneous  flows are allowed to be mixed with the  desired flows on the paths to the terminal, as the  extraneous flows can be cancelled at intermediate nodes or at the terminal.

\subsection{Scalar Time-Invariant Linear Network Mixing}\label{subsec:mixing}

As mentioned in Section~\ref{sec:intro},
to facilitate distributed linear network code designs  for general connections using the mixing concept (without requiring the specific values of local or global coding coefficients   in  the designs), we introduce  local and global mixing variables.  Later, we shall see that distributed linear network mixing designs in terms of these mixing coefficients are much easier. 
Specifically, we introduce  the local mixing  coefficient $\beta_{kij}\in \{0,1\}$ corresponding to edge $(k,i)\in \mathcal E$ and edge $(i,j)\in \mathcal E$,
which relates to the local coding coefficient
$\alpha_{kij}\in \mathcal F$.  Denote $\boldsymbol \beta \triangleq (\beta_{kij})_{(k,i),(i,j)\in\mathcal E}$. 
$\beta_{kij}=1$ indicates that symbol $\sigma_{ki}$ of edge $(k,i)\in \mathcal E$ is allowed  (under our construction)  to contribute to the linear combination over $\mathcal F$ forming symbol $\sigma_{ij}$ in \eqref{eqn:local-coding-coeff} and $\beta_{kij}=0$ otherwise. Thus, if $\beta_{kij}=0$, we have  $\alpha_{kij}=0$;  if $\beta_{kij}=1$, we can further  determine how symbol $\sigma_{ki}$   contributes to  the linear combination  forming symbol $\sigma_{ij}$  by choosing   $\alpha_{kij}\in \mathcal F$  (note that $\alpha_{kij}$ can be zero when $\beta_{kij}=1$).

Similarly, we introduce the global mixing coefficient $x_{ij,p}\in\{0,1\}$ of  flow $p\in \mathcal P$ and edge $(i,j)\in \mathcal E$, which relates to the global coding coefficient $c_{ij,p}\in\mathcal F$. $x_{ij,p}=1$ indicates that flow $p$ is allowed (under our construction)   to be  mixed (coded) with other flows, i.e.,  symbol $\sigma_p$ is allowed to contribute to the linear combination over $\mathcal F$ forming symbol $\sigma_{ij}$ in \eqref{eqn:global-coding-coeff}, and $x_{ij,p}=0$ otherwise. Thus, if $x_{ij,p}=0$, we have $c_{ij,p}=0$;  if $x_{ij,p}=1$, we can further determine how symbol $\sigma_p$ contributes to the linear combination  forming symbol $\sigma_{ij}$ (note that $c_{ij,p}$ can be zero when $x_{ij,p}=1$). Then, we introduce the global mixing vector $\mathbf x_{ij}\triangleq (x_{ij,1},\cdots, x_{ij,p},\cdots, x_{ij,P})\in \{0,1\}^P$  for edge $(i,j)\in \mathcal E$, which relates to the global coding vector $\mathbf c_{ij}= (c_{ij,1},\cdots, c_{ij,p},\cdots,c_{ij,P})\in \mathcal F^P$. Here, $\{0,1\}^P$ represents the set of global mixing vectors, the cardinality of which is $2^P$.

 We consider scalar time-invariant linear network mixing.  In other words, $\beta_{kij}\in \{0,1\}$ and $x_{ij,p}\in \{0,1\}$ are both scalars, and $\beta_{kij}$ and $x_{ij,p}$ do not change over time.



Global mixing vectors provide a natural way of speaking of flows as possibly coded or not without knowledge of the specific values of  global coding vectors. Intuitively,  global mixing vectors can be regarded as a limited  representation  of global coding vectors. Given network mixing vectors, it may not be sufficient to tell  whether a certain symbol can be decoded or not. 
Thus, using the network mixing representation, the extraneous flows, when mixed with the desired flows on the paths to each terminal, are not guaranteed to be cancelled at the terminal.
For decodability to hold at all the terminals, the global mixing vectors at all edges must satisfy the following   feasibility   condition  for scalar linear network mixing.
\begin{Def} [Feasibility of Scalar Linear Network Mixing] For a network $\mathcal G=(\mathcal V, \mathcal E)$ and a set of flows $\mathcal P$ with sources $\mathcal S$  and  terminals $\mathcal T$, a linear network mixing design  $\mathbf x\triangleq (\mathbf x_{ij})_{(i,j)\in\mathcal E}$ 
is called feasible if the following three conditions are satisfied:
 1) $\mathbf x_{s_pj}=\mathbf e_p$ for source edge $(s_p,j)\in \mathcal E$, where $s_p\in \mathcal S$ and $ p\in \mathcal P$; 
2) $\mathbf x_{ij}=\vee_{k\in \mathcal I_i} \beta_{kij}\mathbf x_{ki}$ for edge  $(i,j) \in \mathcal E$ not outgoing from a source, where $i\not\in \mathcal S$ and $\beta_{kij}\in \{0,1\}$;\footnote{Note that $\vee$ denotes the ``or'' operator (logical disjunction).}  
3) $x_{it,p}=0$,  where  $i\in \mathcal I_t$, $p \not\in\mathcal P_t$, $t\in \mathcal T$.

\label{Def:feasibility-mixing}
\end{Def}

Note that Condition 3) in Definition \ref{Def:feasibility-mixing} ensures that for each terminal,  the extraneous flows are not  mixed with the desired flows on the paths to the terminal. In other words,
linear mixing allows only mixing  at intermediate nodes. This is not as general as using linear network coding, which allows mixing and canceling  (i.e., removing one or multiple flows from a mixing of flows) at intermediate nodes.

Given a feasible  linear network mixing design,
one of the ways to implement mixing when $\mathcal F$ is large is to use random linear network coding (RLNC)\cite{Hoetal06,Lun04networkcoding}, 
as discussed in the introduction.   In particular, when $\beta_{kij}=1$, $\alpha_{kij}$ can be randomly, uniformly, and independently chosen in $\mathcal F$ using RLNC; when   $\beta_{kij}=0$,  $\alpha_{kij}$ has to be chosen to be 0.

\section{Mixing Problem Formulation}\label{sec:x-prob-form}

In this section, we formulate  the problem of selecting mixing coefficients  $\boldsymbol \beta$ and  $\mathbf x$ to minimize the cost of the subgraph involved in the coding solution, i.e., the set of edges used in delivering  the flows.

%

In the following formulation, $z_{ij}\in\{0,1\}$  indicates whether edge $(i,j)\in \mathcal E$ is involved in delivering flows, 
and $f_{ij,p}^{t}\in\{0,1\}$ indicates whether edge $(i,j)\in \mathcal E$ is involved in delivering flow $p\in \mathcal P_t$ to terminal $t\in \mathcal T$.  
\begin{Prob} [Mixing]
\begin{align}
U^*(\boldsymbol{\mathcal P})&\triangleq\min_{\mathbf z, \mathbf f, \mathbf x,\boldsymbol{\beta}}\quad  \sum_{(i,j)\in \mathcal E} U_{ij}(z_{ij})\nonumber\\
s.t.\ & z_{ij}\in\{0, 1\}, \  (i,j)\in \mathcal E\label{eqn:z-x-int}\\
&x_{ij,p} \in\mathcal \{0,1\}, \  (i,j)\in \mathcal E,\ p \in \mathcal P\label{eqn:mix-x-int}\\
& \beta_{kij}\in \{0,1\}, \ (k,i), (i,j) \in \mathcal E \label{eqn:mix-beta-int}\\
&f_{ij,p}^{t} \in\{0,1\}, \  (i,j)\in \mathcal E, \ p \in \mathcal P_t, \ t\in \mathcal T\label{eqn:mix-f-int}\\
& \sum_{p\in\mathcal P_t}  f_{ij,p}^{t} \leq z_{ij},  \  (i,j)\in \mathcal E, \ t\in \mathcal T \label{eqn:mix-f-z-int}\\
&\sum_{k\in \mathcal O_i}f_{ik,p}^{t} -\sum_{k\in \mathcal I_i}f_{ki,p}^{t}=\sigma_{i,p}^{t}, \ i \in \mathcal V,  \  p \in \mathcal P_t, \ t\in \mathcal T\label{eqn:mix-f-conv-int}\\
&  f_{ij,p}^{t}\leq x_{ij,p},  \ (i,j)\in \mathcal E, \ p \in \mathcal P_t, \ t\in \mathcal T\label{eqn:mix-f-x-int}\\
& \mathbf x_{s_pj}=\mathbf e_p, \  (s_p,j)\in \mathcal E, \ p\in \mathcal P\label{eqn:f-x-src-int}\\
& \mathbf x_{ij}=\vee_{k\in \mathcal I_i} \beta_{kij}\mathbf x_{ki},   \     (i,j)\in \mathcal E, \ i \not\in\mathcal S\label{eqn:mix-x-inter-int}\\
&x_{it,p}=0, \   i\in \mathcal I_t, \ p \not\in \mathcal P_t, \ t\in \mathcal T\label{eqn:mix-x-dest-int}
\end{align}
where
\begin{align}
\sigma_{i,p}^{t}=
\begin{cases}
1, & i=s_p\\
-1, &i=t\\
0, & \text{otherwise}
\end{cases}\quad  i \in \mathcal V,\ p \in \mathcal P_t, \ t\in \mathcal T.\label{eqn:sigma-1}
\end{align}
 Here, $\mathbf z\triangleq (z_{ij})_{(i,j)\in \mathcal E}$ and $\mathbf f\triangleq (f_{ij,p}^{t})_{(i,j)\in\mathcal E, p\in\mathcal P_t, t\in\mathcal T}$.\footnote{Note that the optimal value in  Problem~\ref{Prob:new-low} is  a function of $\mathcal G$, $\mathcal S$, $\mathcal T$ and $\boldsymbol{\mathcal P}$, which are assumed to be fixed. Here, we write $U^*(\boldsymbol{\mathcal P})$ as a function of $\boldsymbol{\mathcal P}$ only to emphasize the impact of $\boldsymbol{\mathcal P}$ on $U^*(\boldsymbol{\mathcal P})$, which is helpful when considering demand set expansion in Problem~\ref{Prob:new-low-p_t}.}
\label{Prob:new-low}
\end{Prob}

Consider a feasible solution $\mathbf z$, $\mathbf f$, $\mathbf x$ and $\boldsymbol \beta$ to Problem~\ref{Prob:new-low}.
By  \eqref{eqn:mix-f-int} and \eqref{eqn:mix-f-conv-int}, we know that for all $p\in\mathcal P_t$ and $t\in\mathcal T$, all the edges in $\{(i,j)\in \mathcal E: f_{ij,p}^{t}=1\}$ form one  flow path   (i.e., a set of ordered edges $(i,j)\in \mathcal E$ such that $ f_{ij,p}^{t}=1$) from source $s_p$ to terminal $t$. In addition, combining \eqref{eqn:z-x-int} and \eqref{eqn:mix-f-z-int}, we have an equivalent constraint purely in terms of   $\mathbf f$, i.e., 
\begin{align}
& \sum_{p\in\mathcal P_t}  f_{ij,p}^{t}\in \{0,1\},  \quad (i,j)\in \mathcal E, \ t\in \mathcal T. \label{eqn:mix-f-z-c-comb}
\end{align}
 From this,  we know that for all $p,p'\in\mathcal P_t$, $p\neq p'$ and $t\in\mathcal T$, the two flow paths from sources $s_p$ and $s_{p'}$ to terminal $t$ are edge-disjoint. Finally, the feasibility constraints in \eqref{eqn:f-x-src-int}, \eqref{eqn:mix-x-inter-int} and \eqref{eqn:mix-x-dest-int} together with \eqref{eqn:mix-x-int} and \eqref{eqn:mix-beta-int} set other requirements on flow paths (i.e.,  $\mathbf f$) via the constraint in \eqref{eqn:mix-f-x-int}.  Therefore, a feasible solution to Problem~\ref{Prob:new-low} corresponds to a set of flow paths satisfying certain requirements, as illustrated above.
These interpretations can be understood from the following example.
\begin{Exam} [Illustration of Problem~\ref{Prob:new-low}] 
As illustrated in Fig.~\ref{Fig:flow-path}, we consider a network with $\mathcal P=\{1,2\}$, $\mathcal S=\{1,2\}$, $\mathcal T=\{8,7,10\}$, $\mathcal P_1=\{1\}$, $\mathcal P_2=\mathcal P_3=\{1,2\}$, and $U_{ij}(z_{ij})=z_{ij}\in\{0,1\}$ for all $(i,j)\in \mathcal E$. Problem~\ref{Prob:new-low} for the network in Fig.~\ref{Fig:flow-path} has two feasible solutions of network costs 11 and 12, as illustrated in Fig.~\ref{Fig:flow-path} (a) and Fig.~\ref{Fig:flow-path} (b), respectively. Specifically, the two feasible solutions share the same flow paths from sources $s_1$  and $s_2$ to terminals $t_2$ and $t_3$, i.e., flow paths $1-3-4-6-7$, $2-5-7$, $1-3-9-10$, $2-5-4-6-10$.  Notice that the two feasible solutions have different flow paths from source $s_1$ to terminal $t_1$, i.e., flow path $1-3-8$ for the feasible solution in Fig.~\ref{Fig:flow-path} (a) and  flow path $1-3-9-11-8$  for the feasible solution in Fig.~\ref{Fig:flow-path} (b). The optimal solution is the one illustrated in Fig.~\ref{Fig:flow-path} (a) and the optimal network cost is 11.
\end{Exam}


\begin{figure}[t]
\begin{center}
  \subfigure[\small{One feasible solution of network cost 11.}]
  {\resizebox{5.7cm}{!}{\includegraphics{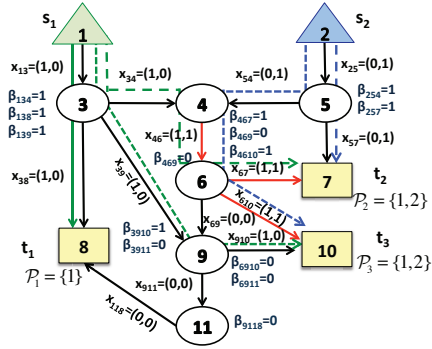}}}
  \subfigure[\small{One feasible solution of network cost 12.}]
  {\resizebox{5.7cm}{!}{\includegraphics{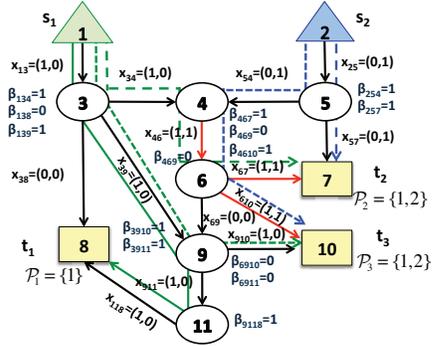}}}
  \end{center}
    \caption{\small{Illustration of feasible solutions to  Problem \ref{Prob:new-low}. $\mathcal P=\{1,2\}$, $\mathcal S=\{1,2\}$, $\mathcal T=\{8, 7,10\}$, $\mathcal P_1=\{1\}$, $\mathcal P_2=\mathcal P_3=\{1,2\}$, $U_{ij}(z_{ij})=z_{ij}\in\{0,1\}$ for all $(i,j)\in \mathcal E$. The flow paths  from the two sources  are illustrated using green and blue curves, respectively. Since  $\mathcal P_2=\mathcal P_3=\{1,2\}$, the flows from $s_1$ to $t_2$ and $s_2$ to $t_3$ are  allowed to be mixed  at edge $(4,6)$. The  red edges carry network-coded information.
In this topology, Problem \ref{Prob:new-low} has two feasible solutions.  However,  neither the two-step mixing approach for general connections in \cite{Lun04networkcoding} nor routing provides a feasible solution.}}
\label{Fig:flow-path}
\end{figure}

 We now illustrate the complexity of Problem~\ref{Prob:new-low}.  The number of variables in $\boldsymbol \beta$ is  $\sum_{(i,j)\in \mathcal E} O_j=\sum_{j\in \mathcal V}I_j O_j\leq D\sum_{j\in \mathcal V}O_j = D E$.   The number of variables in $\mathbf f$ is smaller than or equal to $PTE$.  The numbers of variables in   $\mathbf z$ and  $\mathbf x$ are $E$ and $PE$, respectively. Therefore, the total number of variables in Problem \ref{Prob:new-low} is smaller than or equal to $(D+1)E+(T+1)PE$, i.e., polynomial in $E$, $T$ and $P$. Problem \ref{Prob:new-low} is a binary optimization problem,  and does not appear to have a ready solution.

\begin{Rem} [Problem \ref{Prob:new-low} for Multicast]  When $\mathcal P_t=\mathcal P$ for all $t\in\mathcal T$ (i.e., multicast),  the constraint in \eqref{eqn:mix-x-dest-int} does not exist, and the constraint in \eqref{eqn:mix-f-x-int} is always satisfied by choosing  $\beta_{kij}=1$ for all $(k,i),(i,j)\in \mathcal E$ and choosing  $\mathbf x$ accordingly  by \eqref{eqn:f-x-src-int} and \eqref{eqn:mix-x-inter-int}. Therefore, Problem \ref{Prob:new-low} for general connections reduces to the conventional minimum-cost scalar time-invariant linear network code design problem for the multicast case.  The complexity of the optimization for the multicast case is much lower than that for the general case. This is because  in the optimization for the multicast case, variables $\mathbf x$ and $\boldsymbol \beta$ do not  appear, and
the constraints in \eqref{eqn:mix-x-int}, \eqref{eqn:mix-beta-int}, \eqref{eqn:mix-f-x-int}, \eqref{eqn:f-x-src-int}, \eqref{eqn:mix-x-inter-int} and \eqref{eqn:mix-x-dest-int} can be removed. \label{Rem:spec-case}
\end{Rem}

In the following, we show that a feasible linear network code can be obtained using a feasible solution to Problem \ref{Prob:new-low}  (e.g., using RLNC\cite{Hoetal06}), as illustrated in Section \ref{subsec:mixing}.

\begin{Thm} Suppose Problem \ref{Prob:new-low} is feasible. Then, for each feasible  $\mathbf x$ and  $\boldsymbol \beta$, there exists  a feasible linear network code design  $\boldsymbol \alpha$  and  $\mathbf  c$  with a field size $F>T$ to deliver the desired flows to each terminal.
\label{Thm:feasibility-new-low-opt}
\end{Thm}
\begin{proof} Please refer to Appendix A.
\end{proof}


Next, the minimum network cost of Problem \ref{Prob:new-low} is no greater than the minimum costs of  the two-step mixing approach for general connections in \cite{Lun04networkcoding}  and routing for integer flows,  owing to the following reasons. Problem \ref{Prob:new-low} with $\beta_{kij}=1$ for all  $(k,i), (i,j) \in \mathcal E $, instead of \eqref{eqn:mix-beta-int}, is equivalent to the minimum-cost flow rate control  problem  in the second step of the two-step mixing approach for general connections in \cite{Lun04networkcoding}.  Problem \ref{Prob:new-low} with an extra constraint $\sum_{p\in\mathcal P}x_{ij,p}\in\{0,1\}$ for all $(i,j)\in \mathcal E$ is equivalent to the minimum-cost  routing problem.  Fig. \ref{Fig:flow-path} illustrates a feasible solution to Problem  \ref{Prob:new-low} that cannot be obtained by the two-step mixing approach \cite{Lun04networkcoding} or routing. In this example, the minimum network cost of Problem \ref{Prob:new-low} is smaller than those of the two-step mixing  approach\cite{Lun04networkcoding} and  routing (which can be treated as infinity).

When $\mathcal P_t\cap\mathcal P_{t'}=\emptyset$ for all $t\neq t'$ and $t,t'\in\mathcal T$ (e.g., multiple unicasts), Problem \ref{Prob:new-low} for  general connections reduces to the minimum-cost routing problem and cannot take advantage of the network coding gain.  This is because using the network mixing representation, for decodability to hold, the extraneous flows of each terminal are not allowed to be mixed with the terminal's   desired flows on the path to this terminal, thus limiting the network coding gain.  To address this limitation, we now formulate Problem \ref{Prob:new-low-p_t}, which allows the expansion of the demand sets and the optimization over the expansions to increase the opportunity for mixing flows to different terminals.  Let $\overline{\mathcal P_t}$ denote the expanded demand set, which satisfies $\mathcal P_t\subseteq\overline{\mathcal P_t}\subseteq\mathcal P$. Let $\overline{\boldsymbol{\mathcal P}}\triangleq (\overline{\mathcal P_t})_{t\in\mathcal T}$ denote the expanded demand sets of all the terminals.  \begin{Prob}[Mixing  with Demand Set Expansion]
\begin{align}
U^*=\min_{\{\overline{\mathcal P_t}\}} \quad &U^*(\overline{\boldsymbol {\mathcal P}})\label{eqn:opt-p_t}\\
s.t. \quad & \mathcal P_t\subseteq\overline{\mathcal P_t}\subseteq\mathcal P, \quad t\in\mathcal T\label{eqn:const-p_t}
\end{align}
where $U^*(\overline{\boldsymbol {\mathcal P}})$ is the optimal value to Problem \ref{Prob:new-low} for $\overline{\boldsymbol {\mathcal P}}$.\label{Prob:new-low-p_t}
\end{Prob}

The network coding gain improvement of Problem \ref{Prob:new-low-p_t} can be easily understood from  the case of two unicasts over the butterfly network, as illustrated in Fig.~\ref{Fig:butterfly}. By Theorem~\ref{Thm:feasibility-new-low-opt}, we can easily show the following result.
\begin{Cor} Suppose Problem \ref{Prob:new-low-p_t} is feasible. Then, for each feasible solution, there exists  a feasible linear network code design with a field size $F>T$ to deliver the desired flows to each terminal. \label{Cor:feasibility-new-low-opt-Pt}
\end{Cor}

\begin{figure}[t]
\begin{center}
\includegraphics[height=3.3cm]{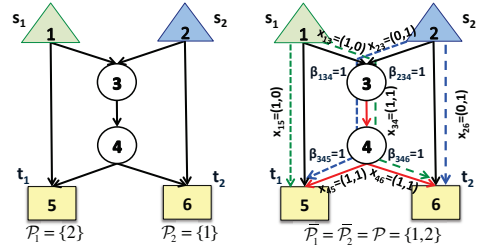}
\caption{\small{Illustration of the network coding gain improvement of Problem \ref{Prob:new-low-p_t} over Problem \ref{Prob:new-low} for two unicasts over the butterfly network.   Problem \ref{Prob:new-low} (left) is not feasible, as the flows from $s_1$ to  $t_2$ and $s_2$  to $t_1$ are not allowed to be mixed over edge $(3,4)$ when $\mathcal P_1\cap\mathcal P_2=\emptyset$. However, Problem \ref{Prob:new-low-p_t} (right) is feasible, as the flows  are allowed to be mixed over edge $(3,4)$ after the demand set expansion to $\mathcal P$.
}}\label{Fig:butterfly}
\end{center}
\end{figure}

By comparing Problem \ref{Prob:new-low} and Problem \ref{Prob:new-low-p_t}, we can  obtain the following lemma.
\begin{Lem}  Suppose  Problem \ref{Prob:new-low} is feasible. Then, Problem \ref{Prob:new-low-p_t} is feasible and
 $U^*\leq U^*(\boldsymbol {\mathcal P})$.
 \label{Lem:low-low-pt}
\end{Lem}

%

In the following sections, we focus on solving Problem \ref{Prob:new-low} for given $ \boldsymbol {\mathcal P}$. However, the obtained centralized and distributed algorithm can be easily extended to solve Problem \ref{Prob:new-low-p_t}  by further optimizing over $\overline{\boldsymbol {\mathcal P}}$. Later, in Section \ref{Sec:simulation}, we shall  illustrate the results for    Problem \ref{Prob:new-low} and Problem \ref{Prob:new-low-p_t}   numerically.

\section{Centralized Algorithm}\label{Sec:centr-alg}

In this section, we develop a centralized algorithm to solve Problem \ref{Prob:new-low},  based on the concept of edge-disjoint flow paths discussed before.  The centralized algorithm is conducted at a central node which is aware of  global network information.

For all $p\in \mathcal{P}_t$ and $t\in \mathcal{T}$, let $N_p^t$ denote the number of flow paths from source $s_p$ to terminal $t$. For Problem \ref{Prob:new-low} to be feasible, assume  $N_p^t>0$.
As illustrated in Section~\ref{sec:x-prob-form}, obtaining a feasible solution to Problem~\ref{Prob:new-low} is equivalent to selecting a set of flow paths   satisfying certain requirements.  Thus, we introduce flow path selection variables to indicate the flow paths selected for information transmission.  
 Let $n_p^t\in\{1,\cdots, N_p^t\}$ denote the flow path selection variable  for source $s_p$ and terminal $t$ (i.e., the index of the selected flow path from source $s_p$ to terminal $t$).
 Denote  the flow path selection variables as $\mathbf n\triangleq (n_p^t)_{ p\in \mathcal{P}_t,  t\in \mathcal{T}}$. In the following,  we express  variables  $\mathbf z, \mathbf f, \mathbf x$ and $\boldsymbol{\beta}$ in terms of $\mathbf n$.   To satisfy   \eqref{eqn:mix-f-int} and \eqref{eqn:mix-f-conv-int}, we require $n_p^t$ to take only one value from $\{1,\cdots, N_p^t\}$. Let  $\mathcal L_p^t(n_p^t)$ denote the selected   flow path (i.e., the set of edges on the selected   flow path) from source $s_p$ to terminal $t$. To satisfy \eqref{eqn:z-x-int} and \eqref{eqn:mix-f-z-int} (or equivalently \eqref{eqn:mix-f-z-c-comb}), we require that the $P_t$  flow paths from sources $\{s_p:p\in\mathcal P_t\}$ to terminal $t$ are edge-disjoint, i.e.,
\begin{align}
\mathcal L_p^t(n_p^t)\cap \mathcal L_{p'}^t(n_{p'}^t)=\emptyset,  \ p, p' \in\mathcal P_t,\  p\neq p',\  t\in \mathcal T.
\label{eqn:path-edge-disjoint-const}
\end{align}
Then, variables $\mathbf f(\mathbf n)\triangleq \left(f_{ij,p}^{t}(n_p^t)\right)_{(i,j)\in \mathcal E, p\in\mathcal P_t, t\in\mathcal T}$  can be expressed in terms of   variables $\mathbf n$  as follows:
\begin{align}
f_{ij,p}^{t}(n_p^t)=\begin{cases}
1,& (i,j)\in \mathcal L_p^t(n_p^t)\\
0,& \text{otherwise}
\end{cases}, (i,j)\in \mathcal E,\ p\in \mathcal{P}_t, \ t\in \mathcal{T}.\label{eqn:path-f}
\end{align}
By \eqref{eqn:mix-f-z-int} and the monotonicity of $U_{ij}(\cdot)$, variables $\mathbf z(\mathbf n)\triangleq \left(z_{ij}(\mathbf n)\right)_{(i,j)\in\mathcal E}$ can be chosen based on  $\mathbf f(\mathbf n)$ and expressed implicitly in terms of variables $\mathbf n$ as follows:
\begin{align}
 z_{ij}(\mathbf n)=\max_{t\in \mathcal T}\sum_{p\in\mathcal P_t}  f_{ij,p}^{t}(n_p^t), \ (i,j)\in \mathcal E.\label{eqn:f-z}
 \end{align}
 In addition, to satisfy \eqref{eqn:mix-x-int}, \eqref{eqn:mix-beta-int}, \eqref{eqn:mix-f-x-int} and \eqref{eqn:mix-x-inter-int},  $\boldsymbol \beta(\mathbf n)\triangleq \beta_{kij}\left(\mathbf n\right)_{(k,i), (i,j)\in \mathcal E}$ can be chosen based on  $\mathbf f(\mathbf n)$ and expressed implicitly in terms of variables $\mathbf n$ as follows:
\begin{align}
\beta_{kij}(\mathbf n)=&\begin{cases}
1, &\max_{t\in\mathcal T, p\in\mathcal P_t}f_{ki,p}^t (n_p^t) f_{ij,p}^t(n_p^t)=1\\
0, & \text{otherwise}
\end{cases},  \nonumber\\
&\hspace{40mm} (k,i), (i,j)\in \mathcal E.\label{eqn:f-beta}
\end{align}
Then, based on  $\boldsymbol \beta(\mathbf n)$,  \eqref{eqn:f-x-src-int} and \eqref{eqn:mix-x-inter-int},  $ \mathbf x(\mathbf n)$ can be determined in topological order\footnote{A topological order of a directed graph $\mathcal G=(\mathcal V, \mathcal E)$ is an ordering of its nodes such that for every directed edge $(i,j)\in\mathcal E$ from node $i\in\mathcal V$ to node $j\in\mathcal V$, $i$ comes before $j$ in the ordering. Such an order exists for the edges of any
directed graph $\mathcal G$ that is acyclic.} and expressed implicitly in terms of variables $\mathbf n$. Finally, to satisfy \eqref{eqn:mix-x-dest-int}, we require
\begin{align}
x_{it,p}(\mathbf n)=0, \   i\in \mathcal I_t, \ p \not\in \mathcal P_t,\ t\in \mathcal T. \label{eqn:mix-x-dest-int-path}
\end{align}
Based on the above relationship between the flow path selection variables  and the variables of  Problem \ref{Prob:new-low}, we now describe the procedure of the centralized algorithm, i.e., Algorithm \ref{alg:new-low-central}, which obtains the feasible flow paths of the minimum network cost.

\begin{algorithm}[t!]
\caption{\small{Centralized Algorithm for Problem \ref{Prob:new-low}}}
\small{\begin{algorithmic}[1]
\STATE For all $p\in \mathcal P_t$ and $t\in \mathcal T$, obtain all the flow paths $\{\mathcal L_p^t(n_p^t):n_p^t\in\{1,\cdots, N_p^t\}\}$ from source $s_p$ to terminal $t$,  using depth-first-search (DFS).
\STATE For all $t\in \mathcal T$, obtain the set of $P_t$ edge-disjoint flow paths  $\mathcal L^t=\{(\mathcal L_p^t(n_p^t))_{p\in\mathcal P_t}: n_p^t\in\{1,\cdots, N_p^t\} \ \text{for all $p\in\mathcal P_t$ and \eqref{eqn:path-edge-disjoint-const} is satisfied}\}$ from  sources  $\{s_p:p\in\mathcal P_t\}$ to terminal $t$.
\STATE Calculate the network costs of  $L=\prod_{t\in\mathcal T}L^t$ combinations of $P_t$ edge-disjoint flow paths for all terminal $t\in\mathcal T$, and sort the $L$ combinations in the ascending order of their network costs, where $L^t=|\mathcal L^t|$ and the $l$-th combination is of the $l$-th smallest network cost $U_l$.
\STATE \textbf{initialize}   $l=1$ and $flag=1$.
 \WHILE{$flag=1$}
 \STATE  For all $p\in \mathcal P_t$ and $t\in \mathcal T$, let $n_p^t$ denote the index of the flow path from source $s_p$ to terminal $t$ in the $l$-th combination.
 \STATE  For all $p\in \mathcal P_t$, $t\in \mathcal T$ and $(i,j)\in \mathcal E$, set $f_{ij,p}^t(n_p^t)$ according to \eqref{eqn:path-f}.
 \STATE  For all $(i,j)\in \mathcal E$, set $z_{ij}(\mathbf n)$ according to \eqref{eqn:f-z}.
 \STATE  For all $(k,i), (i,j) \in \mathcal E $, set $\beta_{kij}(\mathbf n)$ according to \eqref{eqn:f-beta}.
 \STATE  Based on $\{\beta_{kij}(\mathbf n)\}$,  \eqref{eqn:f-x-src-int} and \eqref{eqn:mix-x-inter-int},  determine $ \{ x_{ij,p}(\mathbf n)\}$ in the topological order.
 \IF{\text{\eqref{eqn:mix-x-dest-int-path}  is satisfied}}
 \STATE let $U_x^*(\boldsymbol {\mathcal P})=U_{l}$,  $\mathbf z^*=\mathbf z(\mathbf n)$, $\mathbf f^*=\mathbf f(\mathbf n)\}$, $\mathbf x^*=\mathbf x(\mathbf n)\}$, $\boldsymbol \beta^*=\boldsymbol \beta(\mathbf n)$, and set $flag=0$
 \ELSE
 \STATE set $l=l+1$
 \ENDIF
\ENDWHILE
\end{algorithmic}}\label{alg:new-low-central}
\end{algorithm}

 Note that Constraints \eqref{eqn:mix-f-int} and \eqref{eqn:mix-f-conv-int} are guaranteed in Step~1 and Step~8; Constraints \eqref{eqn:z-x-int} and \eqref{eqn:mix-f-z-int} are guaranteed in Step~2 and Step~7; Constraints \eqref{eqn:mix-x-int}, \eqref{eqn:mix-beta-int}, \eqref{eqn:mix-f-x-int} are guaranteed in Step~9; Constraints \eqref{eqn:f-x-src-int} and \eqref{eqn:mix-x-inter-int} are guaranteed in Step~10; and Constraint \eqref{eqn:mix-x-dest-int} is considered in Steps~11--15. Therefore, we can  see that the  optimization to  Problem \ref{Prob:new-low} can be obtained by Algorithm \ref{alg:new-low-central}.

\section{Probabilistic Distributed Algorithms}\label{sec:dist-mixing}
In this section, using recent results for  CSP \cite{cfl}, we develop two probabilistic distributed algorithms  to solve Problem \ref{Prob:new-low}. 

\subsection{Background on Decentralized CSP}


We review some existing results on  CSP in \cite{cfl}.



\begin{Def} [Constraint Satisfaction Problem] \cite{cfl} A CSP consists of $M$ variables $\{\lambda_1,\cdots, \lambda_M\}$  and $K$ clauses $\{\phi_1,\cdots, \phi_K\}$. Each variable $\lambda_m$ takes values in a finite set $\Lambda$, i.e., $\lambda_m\in \Lambda$ for all $m\in \mathcal M\triangleq \{1,\cdots, M\}$. Let $\boldsymbol{\lambda}\triangleq (\lambda_1,\cdots, \lambda_M)\in \Lambda^M$. Each clause  $k\in \mathcal K\triangleq \{1,\cdots, K\}$ is a function $\phi_k:\Lambda^M \to \{0,1\} $,  where for an assignment of variables $\boldsymbol{\lambda}\in \Lambda^M$, $\phi_k(\boldsymbol{\lambda})=1$ if clause $m$ is satisfied and  $\phi_k(\boldsymbol{\lambda})=0$ otherwise. An assignment $\boldsymbol{\lambda}\in \Lambda^M$ is a solution to the CSP if and only if all clauses are simultaneously satisfied
\begin{align}\min_{k\in \mathcal K}\phi_k(\boldsymbol{\lambda})=1.\label{eqn:CSP-cent}
\end{align}
\label{Def:CSP}
\end{Def}

To solve a CSP in a distributed way, clause participation is introduced  in \cite{cfl}.
Let $\boldsymbol{\lambda}_{-m}\triangleq (\lambda_1,\cdots, \lambda_{m-1},\lambda_{m+1},\cdots, \lambda_M)\in\Lambda^{M-1}$.
For each variable $\lambda_m$, let $\mathcal K_m$ denote the set of clause indices  in which it participates, i.e.,
$\mathcal K_m\triangleq \cup_{\boldsymbol{\lambda}_{-m}\in \Lambda^{M-1}}\{c:\min_{\lambda_m\in \Lambda}\phi_k(\lambda_m,\boldsymbol{\lambda}_{-m})=0,
\max_{\lambda_m\in \Lambda}\phi_k(\lambda_m, \boldsymbol{ \lambda}_{-m})=1\}$. 
Thus, we can rewrite the left hand side of \eqref{eqn:CSP-cent} in a way that focuses on the satisfaction of each variable, i.e., $\min_{m\in \mathcal M}\min_{k\in \mathcal K_m}\phi_k(\boldsymbol{\lambda})=1$. 
This form enables us to solve CSPs in a distributed  iterative  way by locally evaluating the clauses in $\mathcal K_m$ and then updating $\lambda_m$.

CSPs are in general  NP-complete and  most effective  CSP solvers are designed for centralized problems. 
The CFL algorithm \cite[Algorithm 1]{cfl}, summarized  in Algorithm~\ref{alg:cfl}, is a distributed iterative algorithm which can find a satisfying assignment to a CSP almost surely in finite time \cite[Corollary 2]{cfl}. 
Note that Algorithm~\ref{alg:cfl} keeps a probability distribution over all possible values of each variable. The value of each variable is selected from this distribution. For each variable, if all the clauses  in which  a variable participates are satisfied with its current value, the associated  probability distribution is updated to ensure that the variable value remains unchanged;  
if at least one clause is unsatisfied, the probability distribution  evolves  by interpolating between it and a distribution that is uniform on all values except   the one that is currently generating dissatisfaction. Therefore, if all variables are simultaneously satisfied in all clauses, the same assignment of values will be reselected  indefinitely with probability 1.

 \begin{algorithm}[htb]
\caption{\small{Communication-Free Learning \cite{cfl}}}
\small{\begin{algorithmic}[1]
    \STATE Initialize $q_m(\lambda)=\frac{1}{|\Lambda|}$ for all $\lambda\in\Lambda$,  where $|\Lambda|$ denotes the cardinality of $\Lambda$.  
    \LOOP
        \STATE Realize a random variable, selecting $\lambda_m=\lambda$ with probability $q_m(\lambda)$.
        \STATE Evaluate $\min_{k\in \mathcal K_m}\phi_k(\boldsymbol{\lambda})$, returning {\em satisfied} if its value is 1 and {\em unsatisfied} otherwise.
        \IF{{\em satisfied}}
        \STATE set $ q_m(\lambda)=\begin{cases} 1, & \text{if $\lambda=\lambda_m$}\\
        0, & \text{otherwise}
        \end{cases}$
        \ELSE
        \STATE set $q_m(\lambda)=\begin{cases} (1-b)q_m(\lambda)+\frac{a}{|\Lambda|-1+a/b}, & \text{if $\lambda=\lambda_m$}\\
        (1-b)q_m(\lambda)+\frac{b}{|\Lambda|-1+a/b}, & \text{otherwise}
        \end{cases}$, 
        where $a,b\in(0,1]$ are design parameters.
        \ENDIF
    \ENDLOOP
\end{algorithmic}}\label{alg:cfl}
\end{algorithm}

\subsection{Path-based Probabilistic Distributed Algorithm}\label{subset:edge-dist}

In this part, we develop a path-based probabilistic distributed algorithm to solve Problem \ref{Prob:new-low}, using recent results in \cite{cfl}. This distributed algorithm is based on the concept of edge-disjoint flow paths discussed before. It can be viewed as a distributed version of the path-based centralized algorithm, i.e., Algorithm~\ref{alg:new-low-central}. For all $p\in \mathcal P_t$, $t\in\mathcal T$ and $n_p^t\in\mathcal \{1,\cdots, N_p^t\}$, this algorithm requires each node on the $n_p^t$-th flow path from source $s_p$ to terminal $t$ to know its neighboring edge on this flow path. Note that it is not necessary for  each node on the $n_p^t$-th flow path to be aware of other edges on the $n_p^t$-th flow path.

First, we construct a path-based CSP corresponding to the feasibility problem obtained from Problem~\ref{Prob:new-low}. Treat $\mathbf n$ as the variables of the path-based CSP, where $n_p^t\in\{1,\cdots, N_p^t\}$ denotes the index of the selected flow path from source $s_p$ to terminal $t$. As illustrated in Section~\ref{Sec:centr-alg},  the constraints of  $\mathbf f$ in  \eqref{eqn:mix-f-int} and \eqref{eqn:mix-f-conv-int} for Problem~\ref{Prob:new-low} can be  taken into account  by  choosing   $n_p^t\in \{1,\cdots, N_p^t\}$, for all $p\in\mathcal P_t$ and $t\in\mathcal T$. The constraints in \eqref{eqn:z-x-int} and \eqref{eqn:mix-f-z-int} (or equivalently \eqref{eqn:mix-f-z-c-comb}) can be  replaced by  the constraint of $\mathbf n$ in  \eqref{eqn:path-edge-disjoint-const}  for the path-based CSP. Variables $\{\beta_{ijk}(\mathbf n)\}$ and $\{x_{ij,p}(\mathbf n)\}$ can be determined for given $\mathbf n$ via \eqref{eqn:path-f}, \eqref{eqn:f-beta}, \eqref{eqn:f-x-src-int} and \eqref{eqn:mix-x-inter-int}. Thus, the last constraint in \eqref{eqn:mix-x-dest-int} of Problem~\ref{Prob:new-low} can be  replaced by  the  constraint of $\mathbf n$ for the path-based CSP in \eqref{eqn:mix-x-dest-int-path}.
Therefore, we can write the clause for $n_p^t$ as follows:
\begin{align}
\phi_p^{n,t}\left( \mathbf n\right)=\begin{cases}
1, &\text{if  \eqref{eqn:path-edge-disjoint-const}  and \eqref{eqn:mix-x-dest-int-path} hold} \\
0, & \text{otherwise}
\end{cases},\ p\in\mathcal P_t, \ t\in\mathcal T.\label{eqn:phi-n}
\end{align}
We thus have the following proposition.\footnote{Note that the clauses of the path-based CSP cannot be further partitioned, as all the variables $\mathbf n$ are coupled in general.}
\begin{Prop}[Path-based CSP]
The path-based CSP with variables $\mathbf n$ ($n_p^t\in\{1,\cdots, N_p^t\}$) and clauses \eqref{eqn:phi-n}  has considered all the constraints in Problem~\ref{Prob:new-low}.
\end{Prop}

Now, we  present  a path-based distributed probabilistic  algorithm, i.e., Algorithm~\ref{alg:path-cfl},  to obtain a feasible solution to the path-based CSP   using CFL \cite[Algorithm 1]{cfl}. 
Based on the convergence result of CFL \cite[Corollary 2]{cfl}, we know that Algorithm~\ref{alg:path-cfl} can find a feasible solution to Problem~\ref{Prob:new-low} in almost surely finite time. Fig.~\ref{Fig:cov-path-cfl-fig3} illustrates the convergence of Algorithm~\ref{alg:path-cfl}. From Fig.~\ref{Fig:cov-path-cfl-fig3}, we can see that Algorithm~\ref{alg:path-cfl} converges to a feasible solution (i.e., the feasible solution illustrated in Fig.~\ref{Fig:flow-path} (a)) to Problem~\ref{Prob:new-low} for the network in Fig.~\ref{Fig:flow-path} quite quickly (within 35 iterations). This feasible solution corresponds to flow paths $1-3-8$, $1-3-4-6-7$, $2-5-7$, $1-3-9-10$ and $2-5-4-6-10$. The network cost of this feasible solution is 11.

 Relying on  Algorithm~\ref{alg:path-cfl}, we present a path-based distributed probabilistic algorithm, Algorithm~\ref{alg:opt-path},  to obtain the optimal solution to Problem~\ref{Prob:new-low} among multiple feasible solutions obtained by Algorithm~\ref{alg:path-cfl}.\footnote{In Step 3 of Algorithm~\ref{alg:opt-path}, the path-based CFL is run for a sufficiently long time. Step 6 of Algorithm~\ref{alg:opt-path} can be implemented  with a master node obtaining the network cost of the path-based CFL  from all nodes or with all nodes computing the average network cost of the path-based CFL  locally via a gossip algorithm.}  
Since  Algorithm~\ref{alg:path-cfl} can find any feasible solution to Problem~\ref{Prob:new-low} with positive probability, $U_l\to U^*(\{\mathcal P_t\})$ almost surely as $l\to \infty$, where $U_l$ denotes  the minimum network cost obtained by the first $l$ path-based CFLs. Fig.~\ref{Fig:opt-path-fig3} illustrates the convergence of Algorithm~\ref{alg:opt-path}. From Fig.~\ref{Fig:opt-path-fig3}, we can see that Algorithm~\ref{alg:opt-path} obtains the optimal network cost 11 to Problem~\ref{Prob:new-low} for the network in Fig.~\ref{Fig:flow-path}  quite quickly (within 5 iterations).

\begin{algorithm}[t!]
\caption{\small{Path-based CFL}}
\label{alg1}
\small{\begin{algorithmic}[1]
\STATE For all $p\in \mathcal P_t$ and $t\in\mathcal T$, obtain all the flow paths from source $s_p$ to terminal $t$, using DFS.
\STATE For all $p\in \mathcal{P}_t$ and $t\in \mathcal{T}$, source $s_p$ initializes
        ${q_p^t}(n)=\frac{1}{N_p^t}$ for all $n\in \{1,...,N_p^t\}$.
 \LOOP
 \STATE For all $p\in \mathcal{P}_t$ and $t\in \mathcal{T}$, source $s_p$ realizes a random
            variable, selecting $n_p^t=n$
            with probability ${q_p^t}(n)$, where $n\in  \{1,...,N_p^t\}$, and  sends  signaling packet $(p,t, n_p^t)$ over edge $(s_p,j)\in \mathcal L_p^t(n_p^t)$ to node $j$. Once  node $j$
            receives signaling packet $(p,t, n_p^t)$ over edge $(i,j)\in \mathcal L_p^t(n_p^t)$ , it forwards this signaling packet to node $k$ over edge $(j,k)\in \mathcal L_p^t(n_p^t)$.
    \STATE For all $(i,j)\in\mathcal{E}$, if edge $(i,j)$ receives signaling  packets  $\{(p,t, n_p^t): p\in\mathcal P_t\}$ to terminal
            $t\in$ $\mathcal{T}$ from more than one source  in $\{s_p:p \in \mathcal{P}_t\}$, it
            sends NAK back to each of these  sources  along its selected path $n_p^t$.
            \STATE  For all $p\in \mathcal P_t$, $t\in \mathcal T$ and $(i,j)\in \mathcal E$,  set $f_{ij,p}^t(n_p^t)$ according to \eqref{eqn:path-f}.
 \STATE  For all $(k,i), (i,j) \in \mathcal E $,  set $\beta_{kij}(\mathbf n)$ according to \eqref{eqn:f-beta}.
 \STATE  Based on  $\boldsymbol \beta(\mathbf n)$,  \eqref{eqn:f-x-src-int} and \eqref{eqn:mix-x-inter-int},  determine $ \mathbf x(\mathbf n)$ in  topological order.
    \STATE  Every  terminal $t\in\mathcal{T}$ checks \eqref{eqn:mix-x-dest-int-path}. For all $i\in \mathcal I_t$ and $p\not\in\mathcal P_t$, if \eqref{eqn:mix-x-dest-int-path}
 is unsatisfied, terminal $t$   sends NAK back to source $s_{p}$ (along any path to source $s_{p}$), and sends
            NAK back to each of the sources in $\mathcal P_t$ (along its selected path).
    \FOR{$p\in\mathcal P_t$ and $t\in\mathcal T$}
        \IF{source $s_p$ receives no NAKs}
        \STATE set $q_p^t(n)=\left\{
                        \begin{array}{ll}
                            1,\mbox{if } n=n_p^t\\
                            0, \mbox{otherwise} &
                        \end{array}
                    \right.$
        \ELSE
        \STATE  set $q_p^t(n)=
                        \begin{cases}
                            (1-b){q_p^t(n)}+\frac{a}{N_p^t - 1 + \frac{a}{b}},\mbox{if } n=n_p^t\\
                            (1-b){q_p^t(n)}+\frac{b}{N_p^t - 1 + \frac{a}{b}}, \mbox{otherwise} &
                        \end{cases}
                    $,  where $a,b$ $\in$ $(0,1]$ are design parameters.
        \ENDIF
    \ENDFOR
    \ENDLOOP
\end{algorithmic}}\label{alg:path-cfl}
\end{algorithm}

\begin{algorithm}[t!]
\caption{\small{Path-based Distributed  Algorithm}}
\label{alg2}
\small{\begin{algorithmic}[1]
    \STATE $l$ = 1 and $U_1$ = $+\infty$.
    \LOOP
        \STATE Run the path-based CFL in Algorithm~\ref{alg:path-cfl} to the path-based CSP corresponding to Problem~\ref{Prob:new-low}.  Let $\mathbf n_l$ denote the feasible solution obtained by Algorithm~\ref{alg:path-cfl} and let  $\Bar{U}_l$ denote the corresponding network cost.
            \IF{$\Bar{U}_l$ $<$ $U_l$}
                \STATE set $U_{l+1}=\Bar{U}_l$,  $\mathbf n^*=\mathbf n_l$, and $l=l+1$.
            \ENDIF
    \ENDLOOP
\end{algorithmic}}\label{alg:opt-path}
\end{algorithm}

\begin{figure}[t!]
\begin{center}
  \subfigure[\small{Flow path from source $s_1$ to terminal $t_1$.}]
  {\resizebox{9cm}{!}{\includegraphics{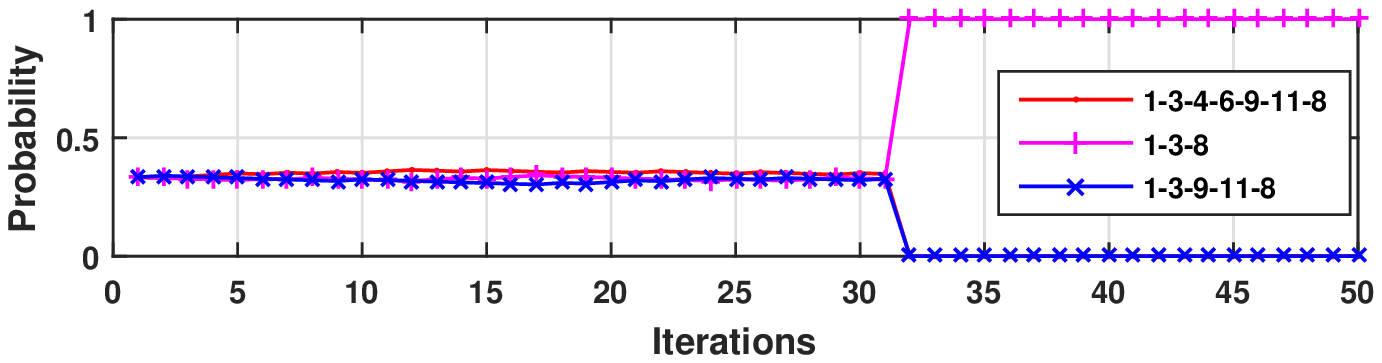}}}
    \subfigure[\small{Flow path from source $s_1$ to terminal $t_2$.}]
  {\resizebox{9cm}{!}{\includegraphics{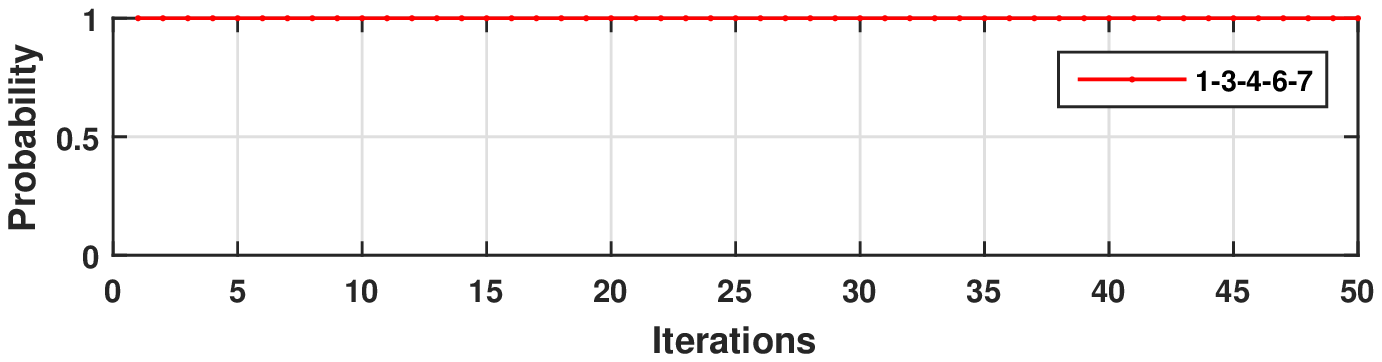}}}
  \subfigure[\small{Flow path from source $s_2$ to terminal $t_2$.}]
  {\resizebox{9cm}{!}{\includegraphics{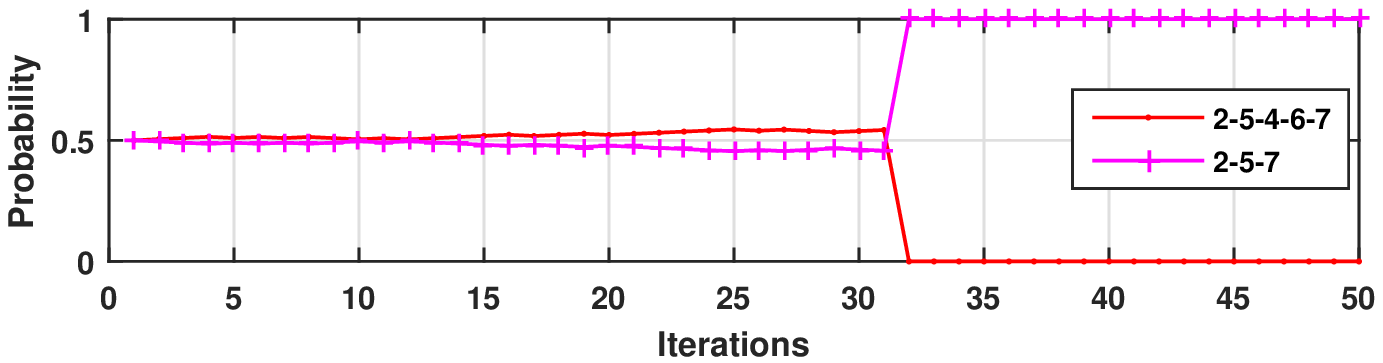}}}
    \subfigure[\small{Flow path from source $s_1$ to terminal $t_3$.}]
  {\resizebox{9cm}{!}{\includegraphics{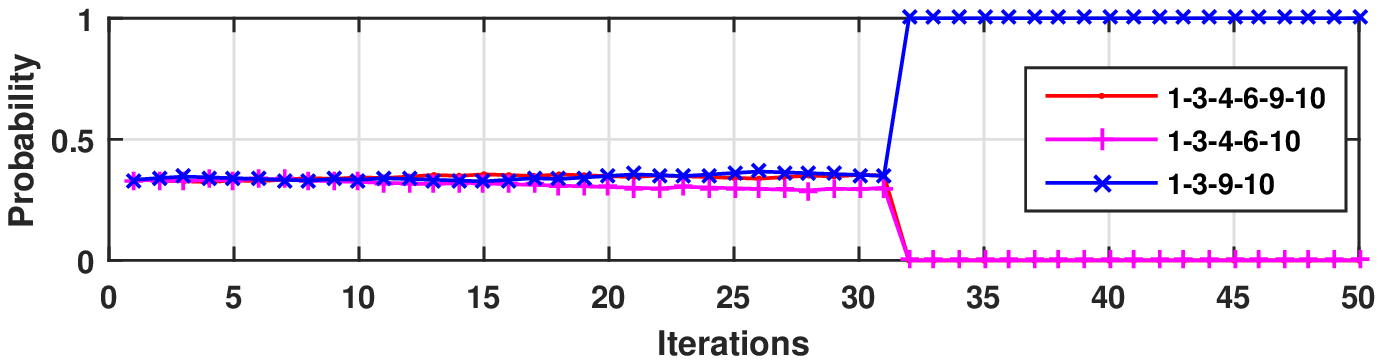}}}
  \subfigure[\small{Flow path from source $s_2$ to terminal $t_3$.}]
  {\resizebox{9cm}{!}{\includegraphics{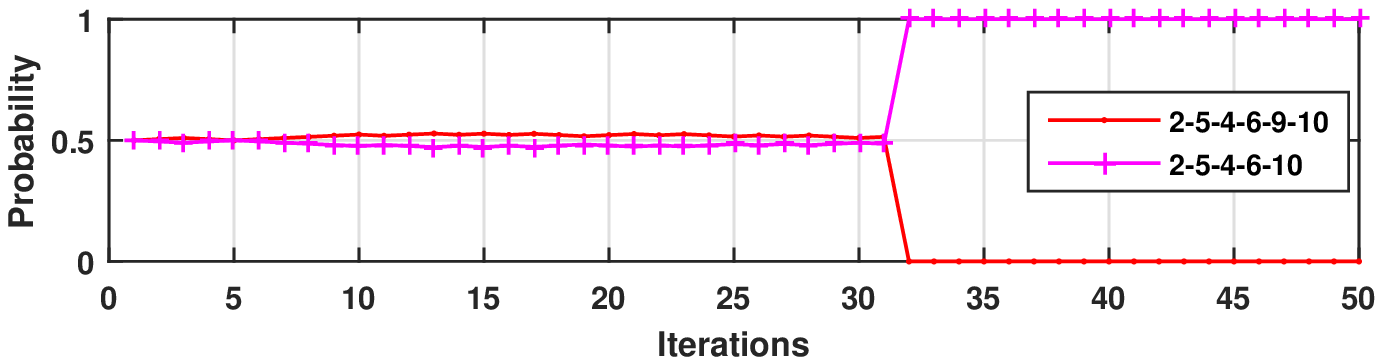}}}
  \end{center}
    \caption{\small{Convergence of the path-based CFL in Algorithm~\ref{alg:path-cfl} for Problem~\ref{Prob:new-low} of the network in Fig.~\ref{Fig:flow-path}. $a=1$ and $b=0.01$.  These convergence curves are for one realization of the random Algorithm~\ref{alg:path-cfl}.} Note that all the flow paths are shown in the figure.}
\label{Fig:cov-path-cfl-fig3}
\end{figure}

\begin{figure}[t!]
\begin{center}
\includegraphics[width=9cm]{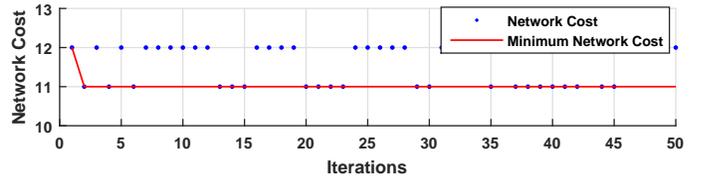}
\caption{\small{Network costs of the path-based CFLs in Algorithm~\ref{alg:opt-path} for Problem~\ref{Prob:new-low} of the network in Fig.~\ref{Fig:flow-path}. Each blue dot represents the network cost of a feasible solution obtained by  the path-based CFL in each iteration of Algorithm~\ref{alg:opt-path}. While the red curve represents the minimum network cost  obtained by Algorithm~\ref{alg:opt-path} within  a certain number of iterations.  The blue dots and red curve are for one realization of the random Algorithm~\ref{alg:opt-path}.}}\label{Fig:opt-path-fig3}
\end{center}
\end{figure}



\subsection{Edge-based Probabilistic Distributed Algorithm}\label{subset:edge-dist}

In this part, we develop an edge-based probabilistic distributed algorithm to solve Problem \ref{Prob:new-low}, using recent results in \cite{cfl}. Compared with the path-based distributed algorithm in Section~\ref{subset:edge-dist},  this edge-based distributed algorithm does not require any path information.

Obtaining a feasible solution to Problem \ref{Prob:new-low}  can be directly treated as a CSP\cite{cfl}.
Specifically,  $\mathbf z,\mathbf f,\mathbf x,\boldsymbol \beta$ and $\{0,1\}$ can be treated as the   variables  and the finite set  of the CSP.  Constraints \eqref{eqn:mix-f-z-int}-\eqref{eqn:mix-x-dest-int} can be treated as the clauses of the CSP. While CSPs are in general NP-complete, several centralized CSP solvers (see references in  \cite{cfl}) and the distributed CSP solver proposed  in \cite{cfl}  can be applied to solve this (na\"{\i}ve) CSP.   However, the direct application of the distributed CSP solver  in \cite{cfl} leads to   high complexity  owing to the large constraint set. In this part, by exploring the features of the constraints in Problem \ref{Prob:new-low}, we  obtain a different  CSP and present  a probabilistic distributed solution with a significantly  reduced number of clauses.

First, we construct a new problem, which we show to be a CSP.  This new problem is better suited than the original problem  to being treated using a probabilistic distributed algorithm  based on the distributed CSP solver presented in  \cite{cfl}. Combining \eqref{eqn:z-x-int} and \eqref{eqn:mix-f-z-int}, we have an equivalent constraint purely in terms of $\mathbf f$, i.e., \eqref{eqn:mix-f-z-c-comb}.
In addition, from \eqref{eqn:mix-x-inter-int}, we have an equivalent constraint purely in terms of  $\mathbf  x$, i.e.,
\begin{align}
\exists \ \beta_{kij}\in \{0,1\} \ \forall   k\in \mathcal I_i,\ & \text{s.t.}\ \mathbf x_{ij}=\vee_{k\in \mathcal I_i} \beta_{kij}\mathbf x_{ki},   \nonumber\\
&   (i,j)\in \mathcal E, \ i \not\in\mathcal S.\label{eqn:mix-x-inter-int-exist}
\end{align}
Therefore, we  can solve only for  variables  $\mathbf f$ and $\mathbf x$ in a distributed way, as
 $\mathbf z$ can be obtained directly from    feasible $\mathbf f$ by choosing $z_{ij}=\max_{t\in \mathcal T}\sum_{p\in\mathcal P_t}  f_{ij,p}^{t}$ according to \eqref{eqn:z-x-int} and \eqref{eqn:mix-f-z-int}, and    $\boldsymbol \beta$ can be obtained from    feasible $\mathbf x$ by \eqref{eqn:f-x-src-int} and \eqref{eqn:mix-x-inter-int}.
We group all the local variables  for each edge $(i,j)\in \mathcal E$ and  introduce the vector variable $(\mathbf f_{ij},\mathbf x_{ij})\in \mathcal Y_{ij}$,
where $\mathbf f_{ij}\triangleq \left( \mathbf f_{ij}^t\right)_{t\in \mathcal T}$,
$\mathbf f_{ij}^{t}\triangleq\left(f_{ij,p}^{t}\right)_{p\in\mathcal P_t}$ and $\mathcal Y_{ij}\triangleq \left\{(\mathbf f_{ij},\mathbf x_{ij}): \eqref{eqn:mix-x-int}, \eqref{eqn:mix-f-int},
\eqref{eqn:mix-f-x-int},\eqref{eqn:f-x-src-int}, \eqref{eqn:mix-x-dest-int}, \eqref{eqn:mix-f-z-c-comb}\right\}$. We also write $\mathcal Y_{ij}=\{\mathbf y_{ij,1},\cdots, \mathbf y_{ij,Y_{ij}}\}$, where $Y_{ij}=|\mathcal Y_{ij}|$.
We now consider a new CSP,  different from the na\"{\i}ve one  that would be directly obtained from Problem~\ref{Prob:new-low}. We treat  $(\mathbf f_{ij},\mathbf x_{ij})$ and $\mathcal Y_{ij}$ as the variable  and   the finite set for edge $(i,j)$ of the CSP.  We write the clauses for $\{(\mathbf f_{ij},\mathbf x_{ij})\}$ as follows:
\begin{align}
&\phi^f_{i}\left( \mathbf f_i \right)=\begin{cases}
1, &\text{if
\eqref{eqn:mix-f-conv-int}  holds $\forall  p\in\mathcal P_t,\ t\in\mathcal T$}\\
0, & \text{otherwise}
\end{cases},\ i\in \mathcal V\label{eqn:phi-f}\\
&\phi^x_{ij}\left( \mathbf x_{ij}, \{\mathbf x_{ki}:k\in\mathcal I_i\} \right) =\begin{cases}
 1, &\text{if  \eqref{eqn:mix-x-inter-int-exist}  holds}\\
0, & \text{otherwise}
\end{cases}, \nonumber\\
&\hspace{50mm}\ (i,j)\in \mathcal E, i\not\in \mathcal S \label{eqn:phi-x}
\end{align}
where $\mathbf f_i\triangleq\left(\mathbf f_{ik}\right)_{k\in\mathcal O_i, k\in\mathcal I_i}$. Note that the local constraints in \eqref{eqn:mix-x-int}, \eqref{eqn:mix-f-int}, 
 \eqref{eqn:mix-f-x-int},   \eqref{eqn:f-x-src-int}, \eqref{eqn:mix-x-dest-int} and \eqref{eqn:mix-f-z-c-comb}  (i.e.,  \eqref{eqn:z-x-int} and \eqref{eqn:mix-f-z-int})  are considered in the finite set $\mathcal Y_{ij}$ of the CSP with respect to each edge $(i,j)\in \mathcal E$. On the other hand, the non-local constraints in \eqref{eqn:mix-f-conv-int} and \eqref{eqn:mix-x-inter-int-exist} are considered in clauses $\phi^f_i$  in \eqref{eqn:phi-f} and  $\phi^x_{ij}$
 in \eqref{eqn:phi-x}, respectively. 
We thus have the following proposition.
 \begin{Prop}[Edge-based CSP] The edge-based CSP with variables $(\mathbf f_{ij},\mathbf x_{ij})\in \mathcal Y_{ij}$, $(i,j)\in\mathcal E$ and clauses \eqref{eqn:phi-f} and \eqref{eqn:phi-x} has considered all the constraints in Problem \ref{Prob:new-low}.
 \end{Prop} 
 
  Note that the number of variables ($E$) and the number of clauses   ($\leq V+E-P$) of the new CSP are much smaller than the number of variables ($\leq (1+D+P+TP)E$) and the number of clauses ($\leq (1+T+TP)E+TPV+TPD$) of the na\"{\i}ve CSP  directly obtained from Problem~\ref{Prob:new-low}.  This  feature  will favor the complexity reduction of  a distributed solution based on  the distributed CSP solver  in \cite{cfl}.

Next, we construct the clause partition. The set of clauses in which variable $(\mathbf f_{ij},\mathbf x_{ij})$ participates is
\begin{align}
\Phi_{ij}=&\left\{\phi^f_{i}, \phi^f_{j}\right\}\cup\left\{\phi^x_{ij},\phi^x_{jk}:i \not\in \mathcal S, k\in\mathcal O_j\right\},\ (i,j)\in \mathcal E.\label{eqn:phi-part}
\end{align} 
Then, the focus can be on the satisfaction of each variable  $(\mathbf f_{ij},\mathbf x_{ij})$, i.e., the satisfaction of each set of clauses $\Phi_{ij}$. Now, the new CSP can be solved using the distributed iterative CFL algorithm \cite[Algorithm 1]{cfl}. Specifically,  each edge $(i,j)\in \mathcal E$ realizes a random variable selecting $(\mathbf f_{ij},\mathbf x_{ij})$. Allow message passing on  $(\mathbf f_{ij},\mathbf x_{ij})$ between adjacent nodes  to evaluate the related clauses. Based on whether the clauses in  \eqref{eqn:phi-part}
 are satisfied or not, the distribution of the random variable of each edge $(i,j)\in \mathcal E$ is updated.   The details are summarized in Algorithm~\ref{alg:edge-cfl},  which obtains a feasible solution to the edge-based CSP using CFL \cite[Algorithm 1]{cfl}.  Based on the convergence 
  result of CFL \cite[Corollary 2]{cfl}, we know that Algorithm~\ref{alg:edge-cfl} can find a feasible solution to Problem~\ref{Prob:new-low} in almost surely finite time.  Fig.~\ref{Fig:cov-edge-cfl-fig3} illustrates the convergence of Algorithm~\ref{alg:edge-cfl}. From Fig.~\ref{Fig:cov-edge-cfl-fig3}, we can see that Algorithm~\ref{alg:edge-cfl} converges to a feasible solution to Problem~\ref{Prob:new-low} for the network in Fig.~\ref{Fig:flow-path} within 5000 iterations. This feasible solution is the same as the one shown in Fig.~\ref{Fig:cov-path-cfl-fig3},  with network cost 11.

 \begin{algorithm}[t!]
\caption{\small{Edge-based CFL}}
\small{\begin{algorithmic}[1]
    \STATE For all $(i,j)\in\mathcal E$, edge $(i,j)$ initializes   $q_{ij}(\mathbf y)=\frac{1}{Y_{ij}}$ for all $\mathbf y \in\mathcal Y_{ij}$.
  \LOOP
        \STATE For all $(i,j)\in\mathcal E$, edge $(i,j)$ realizes a random variable, selecting $(\mathbf f_{ij},\mathbf x_{ij})=\mathbf y$ with probability $q_{ij}(\mathbf y)$, where $\mathbf y\in\mathcal Y_{ij}$.
  \FOR{$(i,j)\in\mathcal E$}
  \STATE Each edge $(i,j)$ evaluates all the clauses in $\Phi_{ij}$.
  \IF{all clauses in $\Phi_{ij}$ are satisfied}
        \STATE set $q_{ij}(\mathbf y)=\begin{cases} 1, & \text{if $\mathbf y=(\mathbf f_{ij},\mathbf x_{ij})$}\\
        0, & \text{otherwise}
        \end{cases}$
        \ELSE
        \STATE set $q_{ij}(\mathbf y)=\begin{cases} (1-b)q_{ij}(\mathbf y)+\frac{a}{Y_{ij}-1+a/b}, & \mathbf y=(\mathbf f_{ij},\mathbf x_{ij})\\
        (1-b)  q_{ij}(\mathbf y)+\frac{b}{Y_{ij}-1+a/b}, & \text{otherwise}
        \end{cases}$, 
        where $a,b\in(0,1]$ are design parameters.
        \ENDIF
                    \ENDFOR
    \ENDLOOP
\end{algorithmic}}\label{alg:edge-cfl}
\end{algorithm}

\begin{algorithm}[t!]
\caption{\small{Edge-based Distributed  Algorithm}}
\small{\begin{algorithmic}[1]
    \STATE $l$ = 1 and $U_1$ = $+\infty$.
    \LOOP
        \STATE Run the edge-based CFL in Algorithm~\ref{alg:edge-cfl} to the edge-based CSP corresponding to Problem~\ref{Prob:new-low}.  Let  $\{(\mathbf f_{ij,l},\mathbf x_{ij,l}):(i,j)\in \mathcal E\}$ denote  the feasible solution obtained by Algorithm~\ref{alg:edge-cfl} and  let $\Bar{U}_l$ denote the corresponding network cost.
            \IF{$\Bar{U}_l$ $<$ $U_l$}
                \STATE set $U_{l+1}=\Bar{U}_l$,  $(\mathbf f_{ij}^*,\mathbf x_{ij}^*)=(\mathbf f_{ij,l},\mathbf x_{ij,l})$ for all $(i,j)\in\mathcal E$, and $l=l+1$.
            \ENDIF
    \ENDLOOP
\end{algorithmic}}\label{alg:opt-edge}
\end{algorithm}

\begin{figure}[t!]
\begin{center}
  \subfigure[\small{Variable $(\mathbf f_{25},\mathbf x_{25})$ for edge $(2,5)$.}]
  {\resizebox{9cm}{!}{\includegraphics{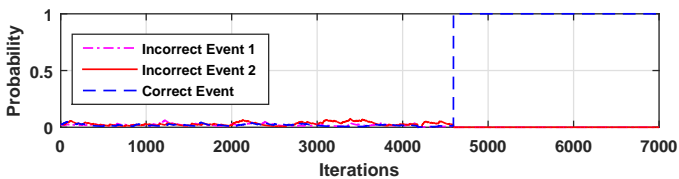}}}
    \subfigure[\small{Variable $(\mathbf f_{57},\mathbf x_{57})$ for edge $(5,7)$.}]
  {\resizebox{9cm}{!}{\includegraphics{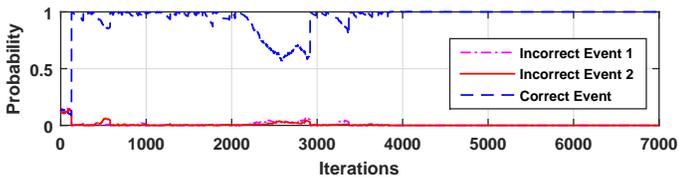}}}
  \subfigure[\small{Variable $(\mathbf f_{67},\mathbf x_{67})$ for edge $(6,7)$.}]
  {\resizebox{9cm}{!}{\includegraphics{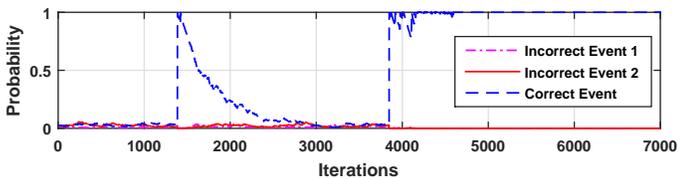}}}
    \end{center}
    \caption{\small{Convergence of the edge-based CFL in Algorithm~\ref{alg:edge-cfl} for Problem~\ref{Prob:new-low} of the network in Fig.~\ref{Fig:flow-path}. $a=1$ and $b=0.01$. Note that for each edge, the ``Correct Event'' indicates the variable taking the value  which corresponds to the feasible solution obtained by the edge-based CFL.  These convergence curves are for one realization of the random Algorithm~\ref{alg:edge-cfl}.}}
\label{Fig:cov-edge-cfl-fig3}
\end{figure}

\begin{figure}[t!]
\begin{center}
\includegraphics[width=9cm]{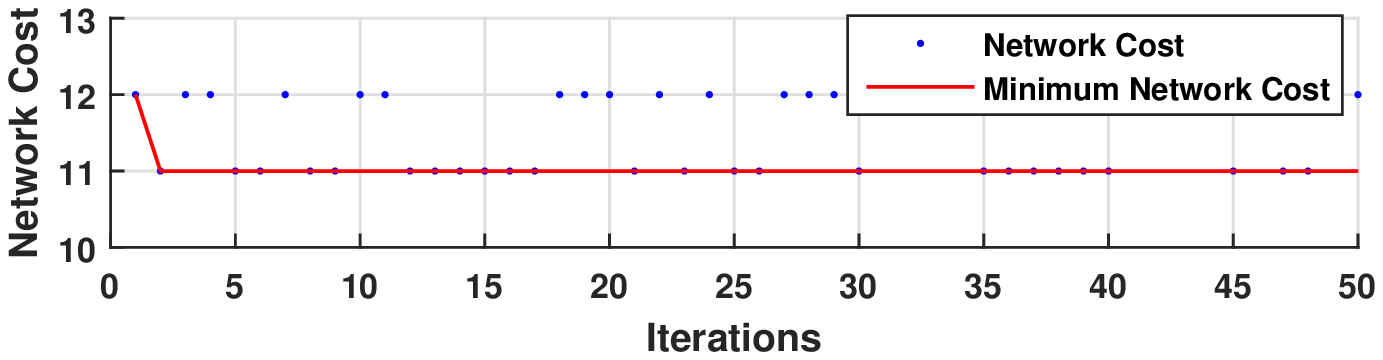}
\caption{\small{Network costs of the edge-based CFLs in Algorithm~\ref{alg:opt-path} for Problem~\ref{Prob:new-low} of the network in Fig.~\ref{Fig:sprint-topo}. Each blue dot represents the network cost of a feasible solution obtained by  the edge-based CFL in each iteration of Algorithm~\ref{alg:opt-path}. While the red curve represents the minimum network cost  obtained by Algorithm~\ref{alg:opt-path} within  a certain number of iterations.  The blue dots and red curve are for one realization of the random Algorithm~\ref{alg:opt-path}.}}\label{Fig:opt-edge-fig3}
\end{center}
\end{figure}

 Relying on  Algorithm~\ref{alg:edge-cfl}, we  present an edge-based distributed probabilistic algorithm, Algorithm~\ref{alg:opt-edge},  to  solve Problem~\ref{Prob:new-low}.\footnote{Note that Step 3 and Step 6 of Algorithm~\ref{alg:opt-edge} can be implemented in similar ways to those in Algorithm~\ref{alg:opt-path}.}   Since  Algorithm~\ref{alg:edge-cfl} can find any feasible solution to Problem~\ref{Prob:new-low} with positive probability,   $U_l\to U^*(\{\mathcal P_t\})$ almost surely as $l\to \infty$, where $U_l$ denotes  the smallest network cost obtained by the first $l$ edge-based CFLs.
Fig.~\ref{Fig:opt-edge-fig3} illustrates the convergence of Algorithm~\ref{alg:opt-edge}. From Fig.~\ref{Fig:opt-edge-fig3}, we can see that Algorithm~\ref{alg:opt-edge} obtains the optimal network cost 11 to Problem~\ref{Prob:new-low} for the network in Fig.~\ref{Fig:flow-path}  quite  quickly (within 5 iterations).

\subsection{Comparison}

In this part, we compare the path-based and edge-based distributed algorithms. In obtaining a feasible solution to Problem~\ref{Prob:new-low}, the path-based CFL, i.e., Algorithm~\ref{alg:path-cfl} and  the edge-based CFL, i.e., Algorithm~\ref{alg:edge-cfl} both base on CFL \cite[Algorithm 1]{cfl}. The convergence result of CFL \cite[Corollary 2]{cfl} guarantee that Algorithm~\ref{alg:path-cfl} and  Algorithm~\ref{alg:edge-cfl} both converge to feasible solutions to Problem~\ref{Prob:new-low} almost surely in finite time. However,  Algorithm~\ref{alg:path-cfl}  converges much faster than Algorithm~\ref{alg:edge-cfl}  in our simulations. This is expected, as Algorithm~\ref{alg:path-cfl}  solves a path-based CSP, while Algorithm~\ref{alg:edge-cfl}    solves an edge-based CSP. The number of variables and the number of possible values for each variable for the path-based CSP are much smaller than those for the edge-based CSP. The difference in the convergence rates of  Algorithm~\ref{alg:path-cfl} and  Algorithm~\ref{alg:edge-cfl}  can be seen by comparing Fig.~\ref{Fig:cov-path-cfl-fig3} and Fig.~\ref{Fig:cov-edge-cfl-fig3}. On the other hand,  Algorithm~\ref{alg:path-cfl}  requires more local information than   Algorithm~\ref{alg:edge-cfl}. In particular, Algorithm~\ref{alg:path-cfl}  requires all the nodes on one path from a source node to a terminal node to be aware of their  neighboring nodes on the path (not all the nodes on the path). Algorithm~\ref{alg:edge-cfl}    instead only requires each node to be aware of its neighboring nodes.

In obtaining  an optimal solution  to Problem~\ref{Prob:new-low} among multiple feasible solutions, the path-based distributed algorithm, i.e., Algorithm~\ref{alg:opt-path} and the edge-based distributed algorithm, i.e., Algorithm~\ref{alg:opt-edge} base on the path-based CFL, i.e., Algorithm~\ref{alg:path-cfl} and  the edge-based CFL, i.e., Algorithm~\ref{alg:edge-cfl}, respectively,  in the same way. Therefore, Algorithm~\ref{alg:opt-path} and Algorithm~\ref{alg:opt-edge} share similar convergence properties. This can be  illustrated in Fig.~\ref{Fig:opt-path-fig3} and Fig.~\ref{Fig:opt-edge-fig3}.

\section{Numerical Illustration}\label{Sec:simulation}

In this section, we numerically  illustrate the  performance of the proposed optimal solutions to Problems \ref{Prob:new-low} and \ref{Prob:new-low-p_t} using mixing only with the two-step mixing approach in \cite{Lun04networkcoding} and optimal routing  for general connections of integer flows.

In the simulation, we consider the Sprint backbone network\cite{sprinttopo:JSAC2011}  as illustrated in Fig. \ref{Fig:sprint-topo}.
We choose sources $\mathcal S=\{8,11\}$ and terminals $\mathcal T\subseteq\{2, 3, 4, 6,9\}$.
The edge directions are chosen to permit connections and help illustrate network coding gain. The green edges have edge cost 1, while the blue edges have edge cost 10 or 20.
The edge costs are chosen to make the network coding advantage  exist at least for some connection requests\cite{KOMT09}. Note that network coding gain takes effect only if transmitting coded information requires a lower network cost   than routing.
We consider 1000 random realizations of demand sets.  For each realization, a pair or triplet of terminals (i.e., $T=2,3$) are  selected from $\{2, 3, 4, 6,9\}$ uniformly at random, and each selected
terminal randomly, uniformly, independently demands a source out of the two sources
in $\mathcal S=\{8,11\}$. 
In addition, each selected terminal randomly chooses to demand the other source or not according to a Bernoulli distribution with  probability $q-1$ of selecting a second source, where $q\in[1,2]$.  Thus, $q$ represents the expected number of sources selected by each terminal (i.e., $P_t$). Note that $q=2$ indicates multicast, and $q=1$ results in unicast connections. In  this way,   general connections are randomly generated with $q$ controlling the average size of the  intersections of the demand sets by different terminals.

\begin{figure}[t!]
\begin{center}
\includegraphics[height=3.7cm, width=7cm]{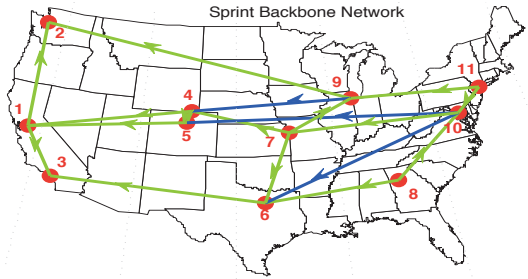}
\caption{\small{Sprint backbone network topology\cite{sprinttopo:JSAC2011}. $\mathcal S=\{8,11\}$ and $\mathcal T\subseteq\{2, 3, 4, 6,9\}$. The edge costs  are: 20 for edges $(10,5)$ and $(10,6)$, 10 for edge $(9,4)$, and 1 for all the other edges.}}\label{Fig:sprint-topo}
\end{center}
\end{figure}


\subsection{Network Cost}
\begin{table}[t!]
\begin{center}
\begin{tabular}{|c|c|c|c|c|}
\hline
& \multicolumn{2}{|c|}{T=2} & \multicolumn{2}{|c|}{T=3} \\
\hline
& q=1.2 & q=1.8 & q=1.2 & q=1.8\\
\hline
{\bfseries{Problem 2}} & 7.49 & 12.70 &12.82 & 20.79 \\
\hline
{\bfseries{Problem 1}} & 8.99 & 13.80 & 16.96&  23.20\\
\hline
{\bfseries{Two-step Mixing}\cite{Lun04networkcoding}} & 8.99 & 13.80 & 17.24 &  24.94\\
\hline
{\bfseries{Routing}} & 9.36 & 18.68 & 17.25 & 32.34 \\
\hline
\end{tabular}
\end{center}
\caption{\small{Average optimal network cost of the network in Fig.~\ref{Fig:sprint-topo}.}}\label{table:cost}
\end{table}

Table. \ref{table:cost} illustrates the average optimal network cost (averaged over 1000 random realizations) for different $q$ and $T$. Note that the optimal network costs of Problems \ref{Prob:new-low} and  \ref{Prob:new-low-p_t} are obtained by the centralized algorithm, i.e., Algorithm~\ref{alg:new-low-central}. We can observe that the average  optimal  network costs of all the schemes increase with increases of $q$  or $T$, i.e., the increase of network load.
The average network costs of the optimal solutions to Problems \ref{Prob:new-low} and  \ref{Prob:new-low-p_t} are lower than the optimal routing, with average cost reductions up to  $28\%$ and $36\%$, respectively.  The average cost reductions are due to the  network coding gain exploited by Problems \ref{Prob:new-low} and  \ref{Prob:new-low-p_t}.
Specifically,  edge $(10,7)$  can serve as the coding edge for the butterfly subnetwork consisting of nodes 6, 7, 8, 9, 10 and 11,  and 
edges $(7,4)$ and $(4,1)$ can serve as the coding edge for the butterfly subnetwork consisting of nodes 1, 2, 3, 4, 6, 7 and 9, 
in  the Sprint backbone network in Fig.~\ref{Fig:sprint-topo}. The network coding gain increases as $q$ or $T$ increases. This is because, using network coding, edges can be used more efficiently in the case of high network load.

In addition,  the average network costs of the optimal solutions to Problems \ref{Prob:new-low} and  \ref{Prob:new-low-p_t} are lower than the two-step mixing approach, with average cost reductions up to $7\%$ and $26\%$, respectively.  The average cost reductions are due to the extra network coding gain (achieved through mixing) exploited by Problems \ref{Prob:new-low} and  \ref{Prob:new-low-p_t}. Specifically, given the demand sets of all the terminals, mixing or not in the two-step mixing approach (determined in the first step, separately from the second flow rate control step)  is restricted by all the physical paths, while mixing or not in Problems \ref{Prob:new-low} and  \ref{Prob:new-low-p_t} (determined jointly with flow rate control) is only restricted by the actual paths that each flow will take, which is also illustrated in the example in Fig. \ref{Fig:flow-path}.  Note that the average network cost  of Problem \ref{Prob:new-low} is lower than the two-step mixing when $T=3$. The average cost reductions of the optimal solutions to Problems \ref{Prob:new-low} and  \ref{Prob:new-low-p_t}  increase as $T$ increases, as there are more physical paths to terminals restricting network coding (mixing) in the two-step mixing approach.

On the other hand, the average network cost of the  optimal solution to Problem \ref{Prob:new-low-p_t} is  lower than that of the optimal solution to  Problem \ref{Prob:new-low}, with average cost reduction up to $24\%$,  illustrating the consequence of  Lemma \ref{Lem:low-low-pt}.
For a given $T$, the performance gain of Problem \ref{Prob:new-low-p_t} over Problem \ref{Prob:new-low} decreases as $q$ increases, since the difference between the feasibility regions of the two problems reduces with the increase of $q$. Note that when $q=2$ (i.e., multicast),  the two problems (feasibility regions) are the same.  However, for a given $q$, the performance gain of Problem \ref{Prob:new-low-p_t} over Problem \ref{Prob:new-low} increases as $T$ increases, since the difference between the feasibility regions of the two problems increases with the increase of $T$. 


\subsection{Convergence}

We  illustrate the convergence performance of own distributed Algorithm \ref{alg:opt-path}. Consider $s_1=8$, $s_2=11$, $t_1=2$, $t_2=6$, $\mathcal P_1=\{1,2\}$, $\mathcal P_2=\{2\}$ and $\mathcal P=\{1,2\}$. In this case, the optimal network costs of Problem \ref{Prob:new-low-p_t}, Problem \ref{Prob:new-low}, the two-step mixing approach  and  routing are 10, 28, 28, 28, respectively.  The optimal network mixing  solution  to Problem \ref{Prob:new-low-p_t} is achieved through the demand set expansion, i.e., $\bar{\mathcal P}_1=\bar{\mathcal P}_2=\mathcal P=\{1,2\}$. The expanded demand set corresponds to multicast, where the network coding gain is  achieved. The optimal mixing (coding) solutions with cost 10 corresponds to flow paths $8-10-7-4-1-2$, $11-10-7-9-2$ ($11-9-2$), $8-6$ and $11-10-7-6$.
There is no network mixing (coding) solution to Problem \ref{Prob:new-low} and the two-step mixing approach.  There are three optimal routing solutions of cost 28, which are also optimal (feasible but non-coding) solutions for Problem~\ref{Prob:new-low} and the two-step mixing approach.  The first one corresponds to flow paths $8-10-5-1-2$, $11-10-7-9-2$ and  $11-10-7-6$. The second one corresponds to flow paths $8-10-7-4-1-2$, $11-9-2$ and $11-10-6$. The third one corresponds to flow paths $8-10-5-1-2$, $11-9-2$ and $11-10-7-6$.   In the following, we illustrate the convergence for the path-based and edge-based distributed algorithms for Problem \ref{Prob:new-low-p_t}  (at $\bar{\mathcal P}_1=\bar{\mathcal P}_2=\mathcal P=\{1,2\}$), respectively.

\subsubsection{Path-based Probabilistic  Distributed Algorithm}

\begin{figure}[t!]
\begin{center}
  \subfigure[\small{Flow path from (source) node 8 to (terminal) node 2.}]
  {\resizebox{9cm}{!}{\includegraphics{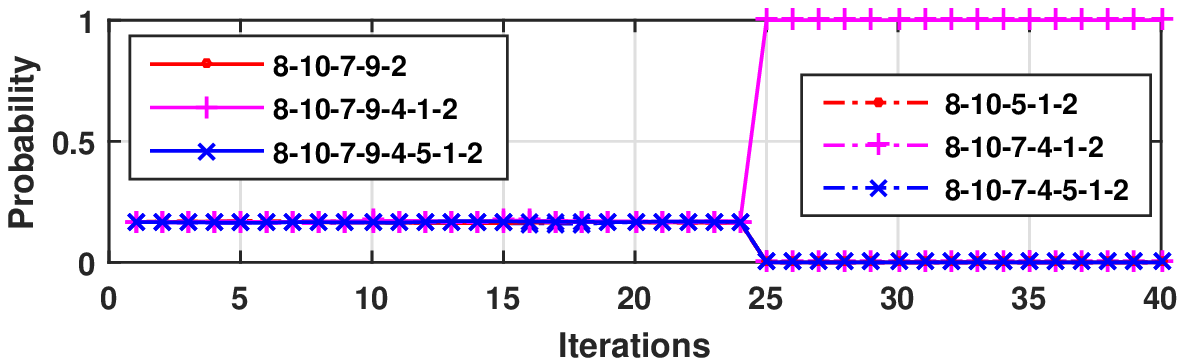}}}
    \subfigure[\small{Flow path from (source) node 11 to (terminal) node 2.}]
  {\resizebox{9cm}{!}{\includegraphics{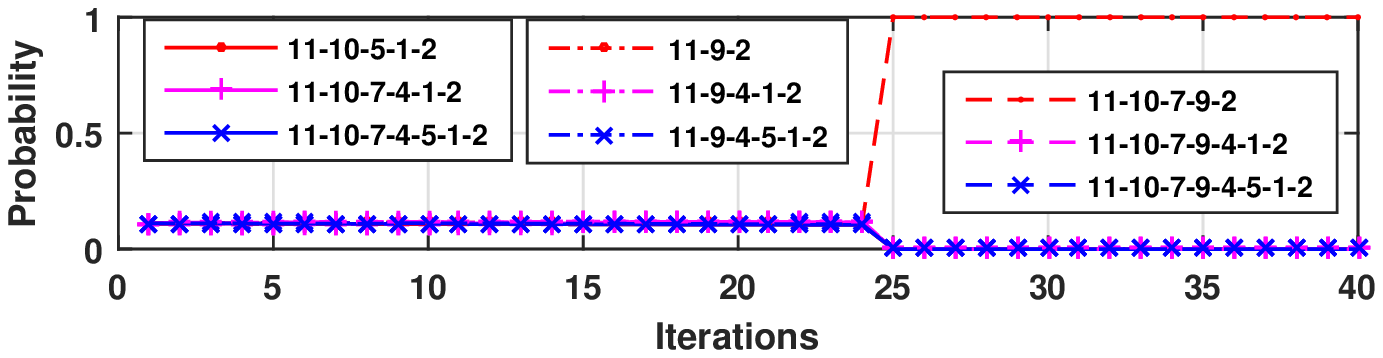}}}
  \subfigure[\small{Flow path from (source) node 8 to   (terminal) node 6.}]
  {\resizebox{9cm}{!}{\includegraphics{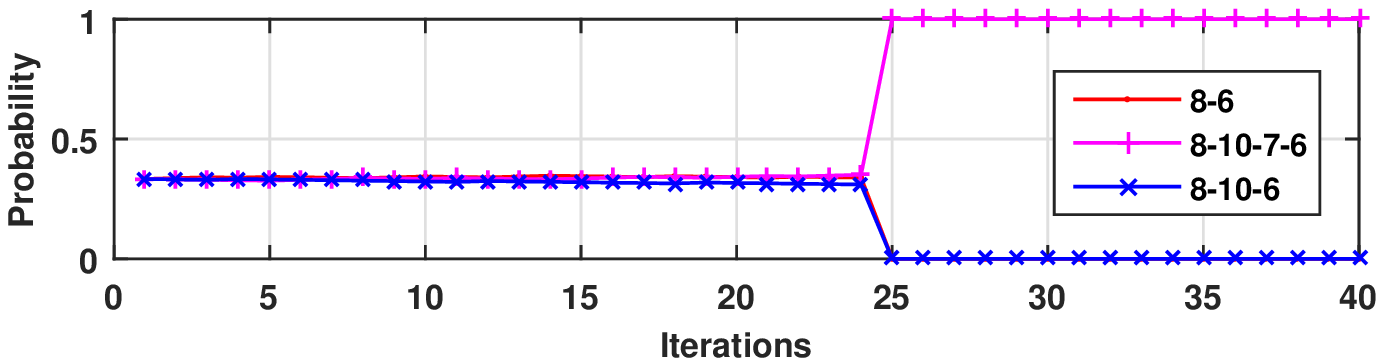}}}
    \subfigure[\small{Flow path from (source) node 11 to   (terminal) node 6.}]
  {\resizebox{9cm}{!}{\includegraphics{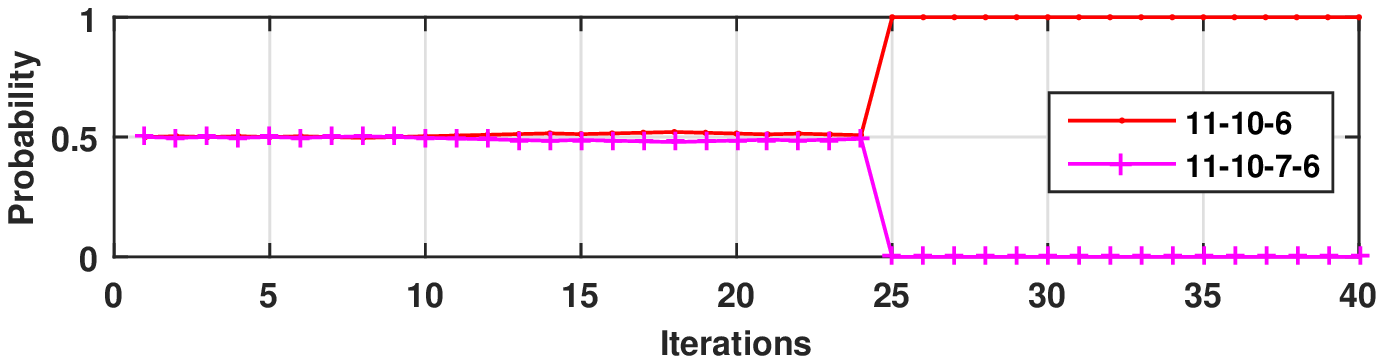}}}
  \end{center}
    \caption{\small{Convergence of the path-based CFL in Algorithm~\ref{alg:path-cfl} for Problem~\ref{Prob:new-low-p_t} of the network in Fig.~\ref{Fig:sprint-topo}. $a = 0.05$ and  $b = 0.009$.  These convergence curves are for one realization of the random Algorithm~\ref{alg:path-cfl}.}  Note that all the flow paths are shown in the figure.}
\label{Fig:cov-path-cfl-sprint}
\end{figure}

\begin{figure}[t!]
\begin{center}
\includegraphics[width=8cm]{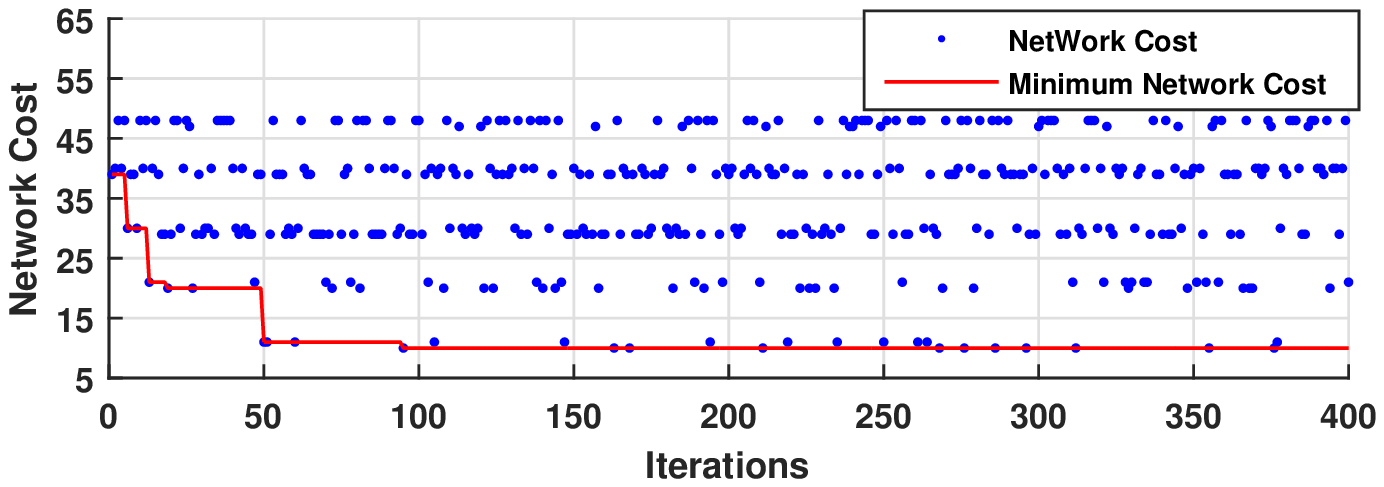}
\caption{\small{Network costs of Algorithm~\ref{alg:opt-path} for Problem~\ref{Prob:new-low-p_t} of the network in Fig.~\ref{Fig:sprint-topo}. Each blue dot represents the network cost of a feasible solution obtained by  the path-based CFL in each iteration of Algorithm~\ref{alg:opt-path}. While the red curve represents the minimum network cost  obtained by Algorithm~\ref{alg:opt-path} within  a certain number of iterations. The blue dots and red curve are  for one realization of the random Algorithm~\ref{alg:opt-path}.}}\label{Fig:opt-path-sprint}
\end{center}
\end{figure}

\begin{figure}[t!]
\begin{center}
\includegraphics[width=7cm]{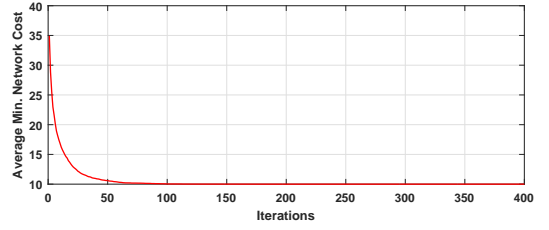}
\caption{\small{Average minimum network costs of the path-based CFLs in Algorithm~\ref{alg:opt-path} for Problem~\ref{Prob:new-low-p_t} of the network in Fig.~\ref{Fig:sprint-topo} over 1000 instances. The red curve here represents the average of the red curves in Fig.~\ref{Fig:opt-path-sprint} over 1000 instances.}}\label{Fig:opt-path-sprint-ave}
\end{center}
\end{figure}

Fig.~\ref{Fig:cov-path-cfl-sprint} illustrates the convergence of Algorithm~\ref{alg:path-cfl} (i.e., Step 3 in  Algorithm~\ref{alg:opt-path}). From Fig.~\ref{Fig:cov-path-cfl-sprint}, we can see that Algorithm~\ref{alg:path-cfl} converges to a feasible solution to Problem \ref{Prob:new-low-p_t}  quite quickly (within 25 iterations). This feasible solution corresponds to flow paths $8-10-7-4-1-2$, $11-10-7-9-2$, $8-10-7-6$ and $11-10-6$ . The network cost of this feasible solution is 10, i.e., the optimal network cost to Problem \ref{Prob:new-low-p_t}. 
Fig.~\ref{Fig:opt-path-sprint} illustrates the convergence of Algorithm~\ref{alg:opt-path} for one  instance. We can see that there exist multiple feasible mixing solutions to Problem~\ref{Prob:new-low-p_t}, which are of different network costs, and  running  Algorithm~\ref{alg:path-cfl}   for multiple times can result in different feasible solutions. Thus, the minimum network cost   may decrease as the number of  iterations   increases.    Algorithm~\ref{alg:opt-path} obtains the optimal network cost 10 to Problem~\ref{Prob:new-low-p_t}    quite  quickly (within 100 iterations). Fig.~\ref{Fig:opt-path-sprint-ave} illustrates the average convergence of Algorithm~\ref{alg:opt-path} over 1000 instances. We can see that on average, within 100 iterations, the minimum network cost under Algorithm~\ref{alg:opt-path} converges to 10, which is the optimal network cost to Problem \ref{Prob:new-low-p_t} obtained by the centralized algorithm in Algorithm~\ref{alg:new-low-central}.

\subsubsection{Edge-based Probabilistic  Distributed Algorithm}

\begin{figure}[t!]
\begin{center}
  \subfigure[\small{Variable $(\mathbf f_{810},\mathbf x_{810})$ for edge $(8,10)$.}]
  {\resizebox{9cm}{!}{\includegraphics{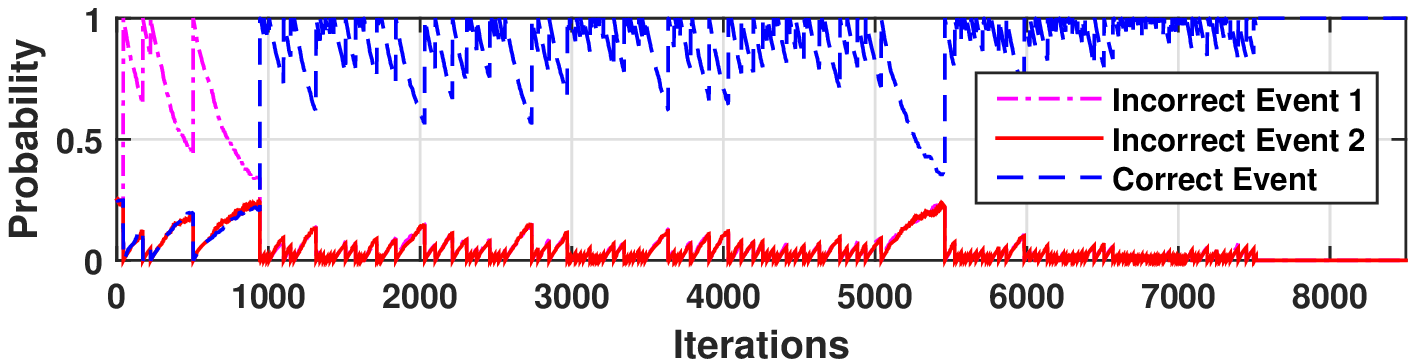}}}
    \subfigure[\small{Variable $(\mathbf f_{1110},\mathbf x_{1110})$ for edge $(11,10)$.}]
  {\resizebox{9cm}{!}{\includegraphics{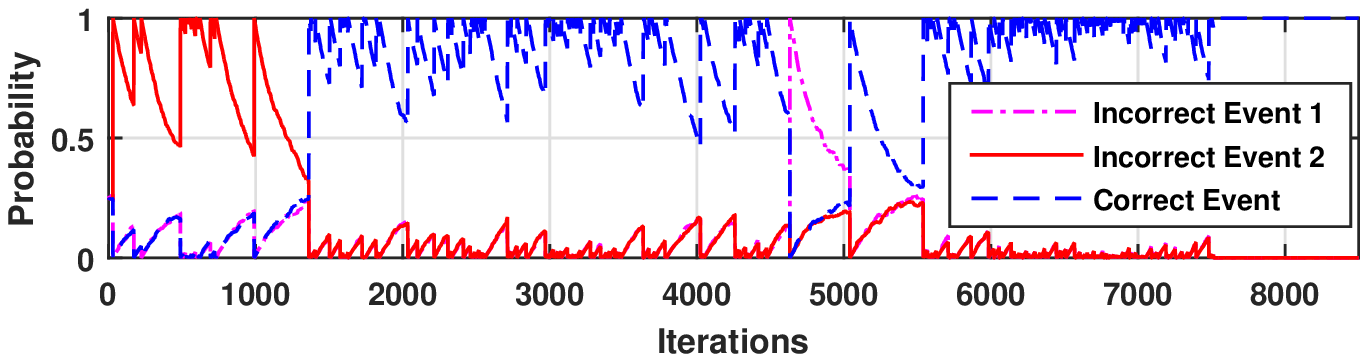}}}
    \subfigure[\small{Variable $(\mathbf f_{12},\mathbf x_{12})$ for edge $(1,2)$.}]
  {\resizebox{9cm}{!}{\includegraphics{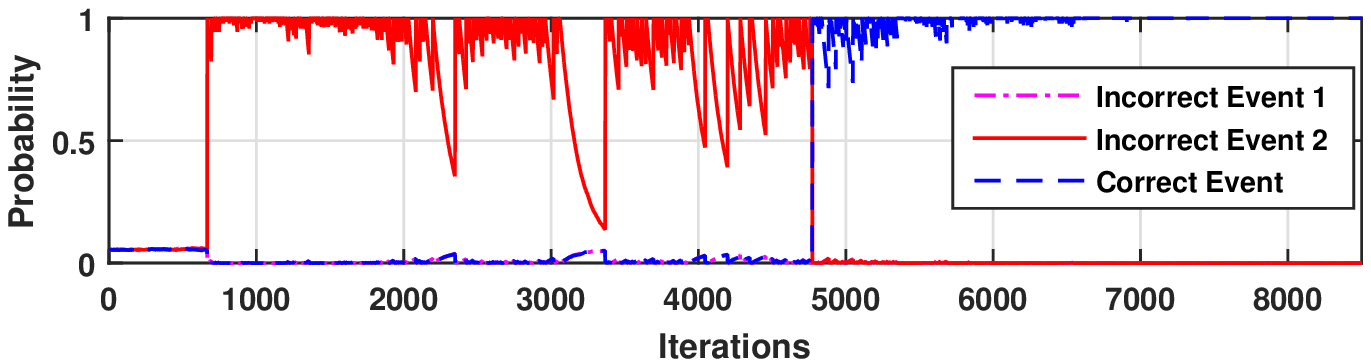}}}
    \end{center}
    \caption{\small{Convergence of the edge-based CFL in Algorithm~\ref{alg:edge-cfl} for Problem~\ref{Prob:new-low-p_t} of the network in Fig.~\ref{Fig:sprint-topo}. $a=1$ and $b=0.01$. Note that for each edge, the ``Correct Event'' indicates the variable taking the value which corresponds to the feasible solution obtained by the edge-based CFL.  These convergence curves are for one realization of the random Algorithm~\ref{alg:edge-cfl}.}}
\label{Fig:cov-edge-cfl-sprint}
\end{figure}

\begin{figure}[t!]
\begin{center}
\includegraphics[width=8cm]{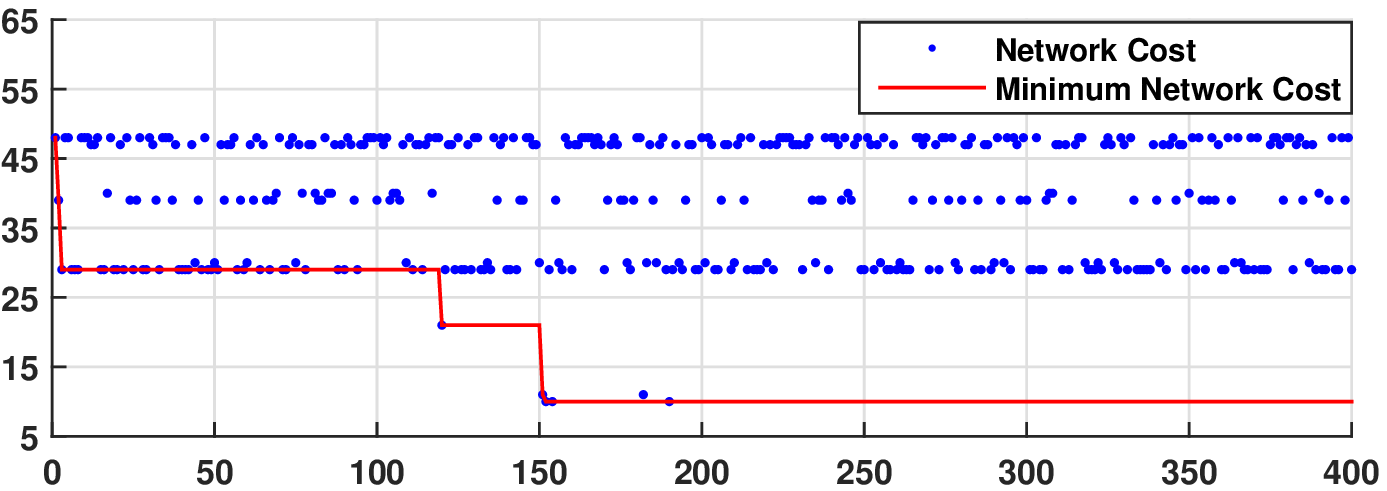}
\caption{\small{Network costs of Algorithm~\ref{alg:opt-edge} for Problem~\ref{Prob:new-low-p_t} of the network in Fig.~\ref{Fig:sprint-topo}. Each blue dot represents the network cost of a feasible solution obtained by  the edge-based CFL in each iteration of Algorithm~\ref{alg:opt-edge}. While the red curve represents the minimum network cost  obtained by Algorithm~\ref{alg:opt-edge} within  a certain number of iterations.  The blue dots and red curve are  for one realization of the random Algorithm~\ref{alg:opt-edge}.}}\label{Fig:opt-edge-sprint}
\end{center}
\end{figure}

\begin{figure}[t!]
\begin{center}
\includegraphics[width=7cm]{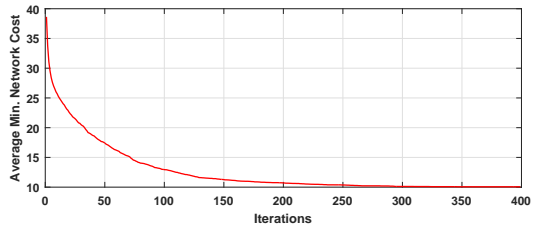}
\caption{\small{Average minimum network costs of the edge-based CFLs in Algorithm~\ref{alg:opt-edge} for Problem~\ref{Prob:new-low-p_t} of the network in Fig.~\ref{Fig:sprint-topo} over 1000 instances. The red curve here represents the average of the red curves in Fig.~\ref{Fig:opt-edge-sprint} over 1000 instances.}}\label{Fig:opt-edge-sprint-ave}
\end{center}
\end{figure}

Fig.~\ref{Fig:cov-edge-cfl-sprint} illustrates the convergence of Algorithm~\ref{alg:edge-cfl} (i.e., Step 3 in  Algorithm~\ref{alg:opt-edge}). From Fig.~\ref{Fig:cov-edge-cfl-sprint}, we can see that Algorithm~\ref{alg:edge-cfl} converges to a feasible solution to Problem \ref{Prob:new-low-p_t} within 8000 iterations. This feasible solution is the same as the one shown in Fig.~\ref{Fig:cov-path-cfl-sprint}, with network cost 10. By comparing Fig.~\ref{Fig:cov-edge-cfl-sprint} with Fig.~\ref{Fig:cov-path-cfl-sprint}, we can see that Algorithm~\ref{alg:edge-cfl} converges much  more  slowly than Algorithm~\ref{alg:path-cfl}. 
Fig.~\ref{Fig:opt-edge-sprint} illustrates the convergence of Algorithm~\ref{alg:opt-edge}. We can see that there exist  for Problem~\ref{Prob:new-low-p_t}, multiple feasible mixing solutions  which  have different network costs, and  running Algorithm~\ref{alg:edge-cfl}    for multiple times can result in different feasible solutions. Thus, the minimum network cost   may decrease as the number of iterations  increases.   Algorithm~\ref{alg:opt-edge} obtains the optimal network cost 10 to Problem~\ref{Prob:new-low-p_t}     within 100 iterations.  Fig.~\ref{Fig:opt-edge-sprint-ave} illustrates the average convergence of Algorithm~\ref{alg:opt-edge} over 1000 instances. We can see that on average, within 300 iterations, the minimum network cost under Algorithm~\ref{alg:opt-edge} converges to 10, which is the optimal network cost to Problem \ref{Prob:new-low-p_t} obtained by the centralized algorithm in Algorithm~\ref{alg:new-low-central}. By comparing Fig.~\ref{Fig:opt-edge-sprint-ave} with Fig.~\ref{Fig:opt-path-sprint-ave}, we can see that Algorithm~\ref{alg:opt-edge} converges much  more slowly than Algorithm~\ref{alg:opt-path}. 

\section{Conclusion and Future Work}

In this paper,  we introduce linear network mixing coefficients  for code constructions of general integer connections. For such code constructions, we pose the problem of cost minimization for the subgraph involved in the coding solution, and relate this minimization to  a  path-based CSP and an edge-based CSP, respectively.  We
present a path-based probabilistic distributed algorithm and an edge-based  probabilistic distributed algorithm with almost sure convergence in finite time  by  applying CFL.
Our approach allows fairly general coding across flows, guarantees no greater cost than routing,    and  demonstrates a possible distributed implementation.  Numerical results illustrate the performance improvement of our approach over  existing methods.

This paper opens up several directions for future research. For instance, the proposed optimization-based linear network code construction for general integer connections can be extended to design route finding protocols of superior performance for general connections.  In addition, a possible direction for future research is to design dynamic approaches not only to build new subgraphs, but also to update them as they  evolve, so as  to reflect changes in topologies for varying networks,  as occur in such settings as peer-to-peer (P2P) networks. Another interesting extension of  the proposed approach to content-centric cache-enabled networks is to incorporate  cache placement (which creates new sources)  into the cost minimization for the subgraph involved in the coding solution in this work. Finally, the proposed approach for wireline networks can also be generalized to wireless networks by considering hyper edges to model broadcast links.

\section*{Appendix A: Proof of Theorem~\ref{Thm:feasibility-new-low-opt}}

Let  $\mathbf z$, $\mathbf x$, $\boldsymbol \beta$  and $\mathbf f$ denote a feasible solution to Problem \ref{Prob:new-low}.  Note that  $\mathbf x$ is uniquely determined by  $\boldsymbol \beta$ according to \eqref{eqn:f-x-src-int} and \eqref{eqn:mix-x-inter-int}, which correspond to Conditions 1) and 2) in Definition \ref{Def:feasibility-mixing}. In addition, by  \eqref{eqn:mix-x-dest-int}, which corresponds to Condition 3) in Definition \ref{Def:feasibility-mixing}, we know that  $\mathbf x$ ensures that for each terminal, the extraneous flows are not mixed with the desired flows on the paths to the terminal. We shall show that,  based on  {$\boldsymbol \beta$}, we can find local coding coefficients  $\boldsymbol \alpha$, which  uniquely determine feasible global coding coefficients  $\mathbf c$ according to
\begin{align}
& \mathbf c_{s_pj}=\mathbf e_p, \quad (s_p,j)\in \mathcal E, \ p\in \mathcal P\label{eqn:f-c-src-proof}\\
&\mathbf c_{ij}=\sum_{k\in\mathcal I_i}\alpha_{kij}\mathbf c_{ki}, \quad (i,j)\in \mathcal E, i\not\in \mathcal S\label{eqn:f-c-edge-proof}
\end{align}
where correspond to   Conditions 1) and 2) in Definition \ref{Def:feasibility}).

First, we choose $\alpha_{kij}=0$ if $\beta_{kij}=0$.  Note that,  as a feasible solution,  $\mathbf x$ is uniquely determined by  $\boldsymbol \beta$ according to \eqref{eqn:f-x-src-int} and \eqref{eqn:mix-x-inter-int}. In addition, we choose  $\mathbf c$ based on   $\boldsymbol \alpha$ according to \eqref{eqn:f-c-src-proof} and \eqref{eqn:f-c-edge-proof}. Thus,  by \eqref{eqn:f-c-src-proof}, \eqref{eqn:f-c-edge-proof}, \eqref{eqn:f-x-src-int} and \eqref{eqn:mix-x-inter-int}, we can show that $c_{ij,p}=0$ if $x_{ij,p}=0$  by induction. Thus, by \eqref{eqn:mix-x-dest-int}, we have
\begin{align}
c_{it,p}=0, \quad  i\in \mathcal I_t, \ p \not\in \mathcal P_t, \ t\in \mathcal T.\label{eqn:mix-c-dest}
\end{align}
In other words, each terminal $t\in \mathcal T$  only needs to consider  $\left(c_{it,p}\right)_{i\in \mathcal I_t, p\in\mathcal P_t}$  for decoding. By \eqref{eqn:mix-f-conv-int}, we can form a flow path from source $s_p$  to terminal $t$, which consists of the edges in  $\mathcal L_p^t\triangleq\{(i,j)\in \mathcal E: f_{ij,p}^{t}=1\}$, where $p\in \mathcal P_t$. By \eqref{eqn:z-x-int} and \eqref{eqn:mix-f-z-int},  we know  that for all $t\in\mathcal T$, there exists $P_t$ edge-disjoint unit flow paths, each one from one source $s_p$ to terminal $t$, where $p\in \mathcal P_t$. Note that  $\mathbf x$ satisfies  all the conditions in Definition \ref{Def:feasibility-mixing}. Thus, by \eqref{eqn:mix-f-x-int}, we know that
all the flow paths satisfy that for each terminal, the extraneous flows (information) are not mixed with the desired flows (information) on the  flow paths to the terminal.
 Let $A_t$ denote the $P_t \times P_t$ matrix, each row (out of $P_t$ rows) of which consists of the $P_t$ elements in  $\left(c_{it,p}\right)_{p\in\mathcal P_t}$ for the last edge $(i,t)$ on one flow path (out of $P_t$ flow paths) to terminal $t$, where $i\in\mathcal I_t$.
Note that  $A_t$ (in terms of  $\left(c_{it,p}\right)_{i\in \mathcal I_t, p\in\mathcal P_t}$ for all   $P_t$ flow paths)  can also be expressed in terms of local coding coefficients $\boldsymbol \alpha$  by \eqref{eqn:f-c-src-proof} and \eqref{eqn:f-c-edge-proof}.\footnote{Given all the local coding coefficients  $\boldsymbol\alpha$, we can compute global coding coefficients $\mathbf c$, and vice versa.} By \eqref{eqn:f-c-src-proof} and \eqref{eqn:f-c-edge-proof}, we know that 1) and 2) of Definition \ref{Def:feasibility} are satisfied. Therefore, it remains to show that 3) of Definition \ref{Def:feasibility} is satisfied. This can be achieved by choosing  $\{ \alpha_{kij}:(k,i),(i,j)\in\mathcal E,\beta_{kij}\neq 0\}$ so that $A_t$ for all $t\in \mathcal T$ are full rank, i.e.,  $\prod_{t\in \mathcal T}\det (A_t)\neq 0$\cite[Pages19-20]{Fragouli:2007:NCF}. 

Next,  we show that if $F>T$, we can choose  $\{ \alpha_{kij}:(k,i),(i,j)\in\mathcal E,\beta_{kij}\neq 0\}$  such that  $\prod_{t\in \mathcal T}\det (A_t)\neq 0$. We first show that  for all $t\in \mathcal T$, $\det (A_t)$  is not  identically equal to zero. For all $p\in\mathcal P_t$ and  $t\in \mathcal T$, choose $\alpha_{kij}=1$ for all edges $(k,i),(i,j)\in \mathcal E$ on the flow path from source $s_p$ to terminal $t$, i.e., $(k,i),(i,j)\in \mathcal L_p^t$, and $\alpha_{kij}=0$ for all edges $(k,i),(i,j)\in \mathcal E$ not on the same flow path, i.e., $(k,i)$ or $(i,j)\not\in \mathcal L_p^t $. This local coding coefficient assignment makes $A_t$ the $P_t\times P_t$ identity matrix. Thus, $\det (A_t)$ is  not identically equal to zero\cite[Page 20]{Fragouli:2007:NCF}.  Then, we show that   $\prod_{t\in \mathcal T}\det (A_t)$  is not   equal to zero, using the algebraic framework in \cite[Pages 31-32]{Fragouli:2007:NCF}. 
 Similarly to the proof of Theorem 3.2 in \cite[Pages 31-32]{Fragouli:2007:NCF}, we can show that  $\prod_{t\in \mathcal T}\det (A_t)$  is a polynomial in unknown variables  $\{ \alpha_{kij}:(k,i),(i,j)\in\mathcal E,\beta_{kij}\neq 0\}$ and that the degree of each unknown variable is at most $T$. Therefore, by Lemma 2.3 \cite[Page 21]{Fragouli:2007:NCF}, we can show  that, for $F>T$, there exists a choice of  $\{ \alpha_{kij}:(k,i),(i,j)\in\mathcal E,\beta_{kij}\neq 0\}$ such that  $\prod_{t\in \mathcal T}\det (A_t)\neq 0$. Recall that $\alpha_{kij}=0$ if $\beta_{kij}=0$. Therefore, based on  $\boldsymbol\beta$, we can obtain  $\boldsymbol\alpha$ that leads  to feasible  $\mathbf  c$.

\end{document}